\documentclass[runningheads]{llncs}
\usepackage{graphicx}
\usepackage{amssymb}
\usepackage{amsmath}

\usepackage{cases}
\usepackage[top=1in, bottom=1in, left=1.25in, right=1.25in]{geometry}
\usepackage{stmaryrd}

    \usepackage{cases}
    \usepackage{pgf}
    \usepackage{tikz}
    \usepackage{framed,multicol}
    \usepackage{enumerate}
    \usepackage{makecell}
    \usepackage{listings}
    \usepackage{multicol}
    \usepackage{wrapfig}
    \usepackage{subcaption}
    \usepackage{hyperref}
    \usepackage{sectsty}

\usepackage{tikz-timing}
\usepackage{tabularx}

\usepackage{mathrsfs}

\usepackage{etex} 

\usepackage{booktabs}

\usepackage{proof}

\usepackage{authblk}

\usetikzlibrary{automata}
\usetikzlibrary{shapes}
\usetikzlibrary{calc,fit}
\usetikzlibrary{arrows,arrows.meta}
\usetikzlibrary{snakes}
\tikzset{line/.style={draw, thick, -latex'}}
\usetikzlibrary{patterns}

\usepackage[linesnumbered,ruled,vlined]{algorithm2e}

\SetCommentSty{mycommfont}

\SetKwInput{KwInput}{Input}                
\SetKwInput{KwOutput}{Output}              
\SetKwRepeat{Do}{do}{while} 

\usepackage{mathtools} 

\usepackage{appendix}

    \newcommand{\rNum}[1]{\expandafter{\romannumeral #1\relax}}
    \newcommand{\rNUM}[1]{\uppercase\expandafter{\romannumeral #1\relax}}

    \newcommand{\mbb}[1]{\mathbb{#1}}
    \newcommand{\mcl}[1]{\mathcal{#1}}
    \newcommand{\spc}[1]{\begin{spacing}{#1}}
    \newcommand{\spce}{\end{spacing}}
    
    \newcommand{\la}[0]{\langle}
    \newcommand{\ra}[0]{\rangle}
    \newcommand{\mfr}[1]{\mathfrak{#1}}

    \DeclareMathOperator{\tran}{\to}

    \newcommand{\incycle}[0]{\unlhd}

    \newcommand{\pincycle}[0]{\lhd}

    \newcommand{\Act}[0]{\mathit{Act}}

    \newcommand{\R}[0]{\mfr{R}}

    \newcommand{\E}[0]{\mathbf{E}}

    \newcommand{\mclul}[1]{\mcl{#1}}
    
    \newcommand{\bo}[0]{\textit{bo}}

    \newcommand{\reg}[0]{\textit{Reg}}
    \newcommand{\LC}[0]{\mcl{L}}
    \newcommand{\eo}[0]{\prec}
    
    \newcommand{\ip}[0]{\textit{in}}

    \newcommand{\Mil}{\textbf{Mil}}
    \newcommand{\MilIfree}{\mbox{$\textbf{BBP}$}}
    \DeclareMathOperator{\Stari}{*}
    \DeclareMathOperator{\Star}{{^\circledast}}
    \newcommand{\assoc}{\textbf{lass}}
    
    \newcommand{\IFSExp}[0]{\mbox{$\mathit{RExp}_{\textit{1-free}}$}}
    \DeclareMathOperator{\IFvdash}{\mbox{$\vdash_{\textit{BBP}}$}}
    \DeclareMathOperator{\lpb}{\curvearrowright}
    \DeclareMathOperator{\subc}{\subseteq_{C}}
    \DeclareMathOperator{\psubc}{\subset_{C}}
    \DeclareMathOperator{\eqvc}{\equiv_{C}}
    \DeclareMathOperator{\eqvn}{\equiv}
    \newcommand{\gra}{\textbf{G}}
    \DeclareMathOperator{\unionc}{\cup_{C}}
    
    \newcommand{\dddef}[0]{=_{\textit{df}}}

        \newcommand{\EII}[0]{\mcl{E}_3}
        \newcommand{\EIIi}[0]{\mcl{E}_{4}}
        \newcommand{\EIIii}[0]{\mcl{E}_{5}}

        \newcommand{\CI}[0]{\mcl{E}_1}
        \newcommand{\CII}[0]{\mcl{E}_2}
        
\makeatletter
\newcommand*{\da@rightarrow}{\mathchar"0\hexnumber@\symAMSa 4B }
\newcommand*{\da@leftarrow}{\mathchar"0\hexnumber@\symAMSa 4C }
\newcommand*{\xdashrightarrow}[2][]{%
  \mathrel{%
    \mathpalette{\da@xarrow{#1}{#2}{}\da@rightarrow{\,}{}}{}%
  }%
}
\newcommand{\xdashleftarrow}[2][]{%
  \mathrel{%
    \mathpalette{\da@xarrow{#1}{#2}\da@leftarrow{}{}{\,}}{}%
  }%
}
\newcommand*{\da@xarrow}[7]{%
  \sbox0{$\ifx#7\scriptstyle\scriptscriptstyle\else\scriptstyle\fi#5#1#6\m@th$}%
  \sbox2{$\ifx#7\scriptstyle\scriptscriptstyle\else\scriptstyle\fi#5#2#6\m@th$}%
  \sbox4{$#7\dabar@\m@th$}%
  \dimen@=\wd0 %
  \ifdim\wd2 >\dimen@
    \dimen@=\wd2 %
  \fi
  \count@=2 %
  \def\da@bars{\dabar@\dabar@}%
  \@whiledim\count@\wd4<\dimen@\do{%
    \advance\count@\@ne
    \expandafter\def\expandafter\da@bars\expandafter{%
      \da@bars
      \dabar@
    }%
  }%
  \mathrel{#3}%
  \mathrel{%
    \mathop{\da@bars}\limits
    \ifx\\#1\\%
    \else
      _{\copy0}%
    \fi
    \ifx\\#2\\%
    \else
      ^{\copy2}%
    \fi
  }%
  \mathrel{#4}%
}
\makeatother

\begin{document}
\title{Image Reflection on Process Graphs --- A Novel Approach for the Completeness of an Axiomatization of 1-Free Regular Expressions Modulo Bisimilarity}
%
%

\author{Yuanrui Zhang\inst{1}\orcidID{0000-0002-0685-6905} \and
Xinxin Liu\inst{2}\orcidID{0000-0002-8334-8277}}

\authorrunning{Y. Zhang and X. Liu}

%

\institute{School of Mathematics and Statistics, Southwest University, China\\
\email{zhangyrmath@126.com, zhangyrmath@swu.edu.cn} \\
\and
Institute of Software,
Chinese Academy of Sciences, China\\
\email{xinxin@ios.ac.cn}
}

\maketitle              
\begin{abstract}
We analyze a phenomenon called ``image reflection'' on a type of characterization graphs --- LLEE charts --- of 1-free regular expressions. 
Due to the correspondence between 1-free regular expressions and the provable solutions of LEE/LLEE charts,  
this observation naturally leads to a new proof for the completeness of the proof system \MilIfree\ for 1-free regular expressions modulo bisimulation equivalence. 
The critical part of the previous proof is to show that bisimulation collapse, which plays the role in linking the provable solutions of two LLEE charts, is still an LLEE chart. 
The difference of our proof, compared to the previous one, is that  
we do not rely on the graph transformations from LLEE charts into their bisimulation collapses by merging two  carefully-selected bisimilar nodes in each transformation step. 
Instead, we directly show that the bisimulation collapse of an LLEE chart possesses an LEE/LLEE structure based on its set of images 
mapped through the bisimulation function from the LLEE chart,
and the constrained relation between the images 
and their so-called ``well-structured'' looping-back charts pre-images 
on the LLEE chart. 
Our approach provides a novel angle to look at this problem and related problems, and can also be used for simplifying the graph transformations in the proof of the completeness problem of the proof system \Mil\ for regular expressions modulo bisimulation equivalence, which had remained open until very recently.


\end{abstract}

\section{Introduction}
\label{section:Introduction}

Regular expression, proposed by Kleene~\cite{Kleene+1956+3+42}, is a formal language consisting of elements 
$a$ (\emph{actions}), $0$ (\emph{deadlock}), $1$ (\emph{skip}) and their compositions $e_1 + e_2$ (\emph{non-deterministic choice}), $e_1 \cdot e_2$ (\emph{sequence}) and $e^*$ (\emph{iterations}). 
The original star operator was a binary one: $e_1 \Stari e_2$ (meaning iterating $e_1$ for finite times and then executing $e_2$, or iterating $e_1$ infinitely). 
Later Copi et al. in their work~\cite{Copi1958} proposed to use the unary operator $e^*$ instead of $e_1\Stari e_2$. 
Regular expression has many applications in the field of computer science, 
especially as a formalism or when embedded into a formalism such as Kleene algebra with test~\cite{Kozen97}, dynamic logic~\cite{Harel00}, etc, it can serve for the specification and verification purposes of different systems (cf. work like~\cite{Smolka20,Beckert2016,Platzer2018,Zhang21,Zhang22}). 

Equational axiomatisation of regular expressions is a basic problem.
It asks for 
a set of equational axioms and inference rules (together forming a \emph{proof system}) which are semantically correct (namely the proof system's \emph{soundness}), and based on which other semantically-correct equations of regular expressions can be derived. 
When all semantically-correct equations are derivable, a proof system is called \emph{complete}. 
Salomaa~\cite{Salomaa1966} gave a complete axiomatisation for regular expressions under the semantics of execution traces, 
where the behaviour of a regular expression is interpreted as a set of traces of actions. 
Milner first studied the axiomatisation of regular expressions under the \emph{process semantics}. 
In this semantics, the behaviour of a regular expression $e$ is interpreted as a set of   \emph{bisimilar charts}, instead of execution traces. 
In~\cite{Milner1984}, Milner built
a sound proof system \Mil\ for regular expressions (there he called regular expression ``star expression''). 
\Mil\ is tailored from Solomaa's proof system by removing the critical left associative law (\assoc): $e_3\cdot (e_1 + e_2) = e_3\cdot e_1 + e_3\cdot e_2$, which does not hold under the process semantics anymore.  
It was claimed that the completeness problem of \Mil\ could be much harder because the Solomaa's proof approach, which relies on (\assoc), cannot be applied.  
This work by Milner has prompted a series of research work since then. 
Up to now, some partial solutions and a recent full solution have been acquired, considering various subsets of regular expressions together with their adapted proof systems from \Mil:

\cite{Bergstra94}
proposed an axiomatisation for a subset of regular expressions where 
there are no $0$ and $1$, but a binary star operator $e_1\Stari e_2$ instead of $e^*$. 
Its completeness was proved in~\cite{Fokkink96}. 
\cite{Fokkink1996AnAF} built a complete proof system 
for a sub-regular expression with $0$ and $1$, but with a ``terminal cycle'' $e^\omega$ rather than $e^*$. $e^\omega$ is semantically equivalent to $e^* \cdot 0$ in regular expression. 
Recently, Grabmayer and Fokkink in~\cite{Grabmayer20} axiomatized the so-called \emph{1-free regular expressions}, in which there are $0$ and a binary star operator $e_1\Stari e_2$ (denoted by $e_1\Star e_2$ in~\cite{Grabmayer20}) instead of $e^*$, but no $1$. 
The proof system tailored from \Mil\ for 1-free regular expressions is denoted by \MilIfree. 
In order to prove the completeness of \MilIfree, 
the authors developed a characterization graph called \emph{LEE/LLEE charts} and adopted a nominal ``minimization strategy'' (introduced below) which relies on a series of LLEE-structure-preservable graph transformations to minimize LLEE charts.  
The critical LEE/LLEE structure developed in~\cite{Grabmayer20} was further extended to the case for regular expressions~\cite{Grabmayer20-2}. 
Based on it a proof sketch was given for the completeness of \Mil~\cite{Grabmayer2022}. 

\ifx
Based on it a method to prove the completeness of \Mil\ was fully explored in~\cite{Grabmayer2022}, where the author claimed that this open problem has finally been solved. 
\fi

In this paper, we propose a novel approach for the completeness of \MilIfree --- as a novel partial solution for the completeness of $\Mil$.  
We adopt a different proof strategy: to directly prove that the 
minimization graph (also called \emph{bisimulation collapse}) $\mcl{H}$ of an LLEE chart $\mcl{G}$ is an LLEE chart, without performing the LLEE-structure-preservable graph transformations from $\mcl{G}$ to $\mcl{H}$ as in~\cite{Grabmayer20}. 
The idea is based on an observation that, 
the behaviour of each image on $\mcl{H}$ of the bisimulation function $\theta$ from $\mcl{G}$ to $\mcl{H}$ is constrained by a so-called \emph{well-structured looping-back chart}, so that 
all of the images on $\mcl{H}$ form a hierarchy structure that satisfies the LEE property. 
\ifx
an LLEE sub-structure $\LC$ (called a \emph{looping-back chart}) on $\mcl{G}$ through the bisimulation function from $\mcl{G}$ to $\mcl{H}$ is constrained by the behaviour of $\LC$, so that we can `structure' all of the images on $\mcl{H}$ in a way 
such that $\mcl{H}$ satisfies the LEE property. 
\fi
We call this phenomenon the \emph{image reflection} on $\mcl{G}$ by $\mcl{H}$ via $\theta$.  
Below in this paper, we will make this ``hierarchy structure'' explicit. 
Compared to~\cite{Grabmayer20}, 
our approach simplifies the whole proof process by skipping the 
so-called ``connect $x_1$ through-to $x_2$'' graph transformations, where 
two bisimilar nodes $x_1$ and $x_2$ need to be carefully selected and merged each time (cf. Proposition 6.4 of~\cite{Grabmayer20} for the selection rules). 
More importantly, 
our approach 
provides a novel angle to look at this problem and related problems, and is
also useful for simplifying the current proof of the completeness problem of \Mil\ by~\cite{Grabmayer2022}.  


In the following, we briefly introduce the minimization strategy and our idea. 
The concepts and notations below will be formally defined in Sect.~\ref{section:Preliminaries} and~\ref{section:Looping-back Charts and well-structured Images}.

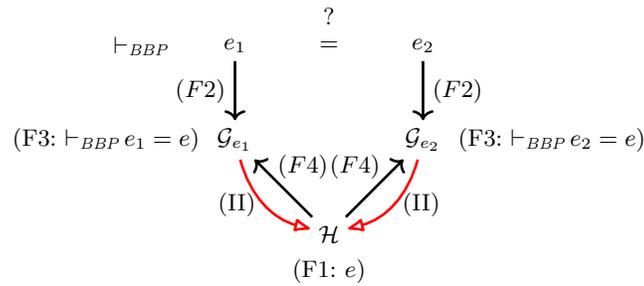
\begin{figure}[tbhp]
    \centering
    \noindent\makebox[\textwidth]{ 
    \begin{tabular}{c}
    \begin{tikzpicture}[->,>=stealth', node distance=1.7cm]

     \coordinate (H) at (0,0);
     \coordinate (Gi) at ($(H) + (-1.25cm, 1.25cm)$);
     \coordinate (Gii) at ($(H) + (1.25cm, 1.25cm)$);
     \coordinate (ei) at ($(Gi) + (0,1.25cm)$);
     \coordinate (eii) at ($(Gii) + (0,1.25cm)$);
     \coordinate (eq) at ($(ei) + (1.25cm,0)$);
     \coordinate (der) at ($(ei) + (-1.25cm,0)$);
     
    \node[draw=none, label=south:\mbox{\footnotesize{(F1: $e$)}}] at (H) (nH) {$\mcl{H}$};
    \node[draw=none, label=west:\mbox{\footnotesize{(F3: $\IFvdash e_1 = e$)}}] at (Gi) (nGi) {$\mcl{G}_{e_1}$};
    \node[draw=none, label=east:\mbox{\footnotesize{(F3: $\IFvdash e_2 = e$)}}] at (Gii) (nGii) {$\mcl{G}_{e_2}$}; 
    \node[draw=none] at (ei) (nei) {$e_1$};
    \node[draw=none] at (eii) (neii) {$e_2$};
    \node[draw=none, label = north:$?$] at (eq) (neq) {$=$};
    \node[draw=none] at (der) (nder) {$\IFvdash$};
    
     \path
     (nei) edge [-To, line width=1pt] node [left,black] {\footnotesize{$(F2)$}} (nGi)
     (neii) edge [-To, line width=1pt] node [right,black] {\footnotesize{$(F2)$}} (nGii)
     (nGi) edge [-{Latex[round,open]}, line width=1pt, bend right, red] node [left,black] {\footnotesize{(II)}} (nH)
     (nGii) edge [-{Latex[round,open]}, line width=1pt, bend left, red] node [right,black] {\footnotesize{(II)}} (nH)
     (nH) edge [-To, line width=1pt] node [above right,black, xshift=-0.2cm] {\footnotesize{$(F4)$}} (nGi)
     (nH) edge [-To, line width=1pt] node [above left,black, xshift=0.2cm] {\footnotesize{$(F4)$}} (nGii)
    ;



    \path
     ;

    


    
    \end{tikzpicture}
    \\
    \end{tabular}
    } 
        \caption{Minimization strategy}
        \label{figure:Minimization strategy}
\end{figure}

The minimization strategy (cf.~\cite{Grabmayer20}) is based on the following facts: 
\begin{enumerate}
\item[(F1)]
Each LLEE chart $\mcl{G}$ corresponds to a set of equations that has a 1-free expression as a \emph{primary provable solution} (simply \emph{primary solution}). 

\item[(F2)]
Each 1-free regular expression $e$ corresponds to an LLEE chart, namely $\mcl{G}_e$, which captures its process semantics; 

\item[(F3)]
For any two primary solutions $e$ and $e'$ of an LLEE chart $\mcl{G}$, 
equation $e = e'$ is derivable. 

\item[(F4)]
For any bisimulation function $\theta$ from $\mcl{G}$ to $\mcl{H}$, each primary solution of $\mcl{H}$ is also a primary solution of $\mcl{G}$. 
\end{enumerate}

(F4) was also applied to solve the completeness of axiomatisations by~\cite{Liu2020}. 

The idea of minimization strategy (illustrated in Fig.~\ref{figure:Minimization strategy}) can be stated as follows on two stages:
\begin{enumerate}
\item[(I)]
To prove equation $e_1 = e_2$, 
by (F2), 
this means there are two LLEE charts $\mcl{G}_{e_1}$ and $\mcl{G}_{e_2}$ 
bisimilar to each other (denoted by $\mcl{G}_{e_1}\sim \mcl{G}_{e_2}$) that have $e_1$ and $e_2$ as their primary solutions respectively. 
Due to $\mcl{G}_{e_1}\sim \mcl{G}_{e_2}$, $\mcl{G}_{e_1}$ and $\mcl{G}_{e_2}$ have the common bisimulation collapse $\mcl{H}$ through two bisimulation functions
$\theta_1$ from $\mcl{G}_{e_1}$ to $\mcl{H}$ and $\theta_2$ from $\mcl{G}_{e_2}$ to $\mcl{H}$ respectively. 
Assume that $\mcl{H}$ is an LLEE chart. 
Then by (F1) $\mcl{H}$ has a primary solution, let us say $e$. 
By (F4), $e$ is also a primary solution of $\mcl{G}_{e_1}$ and $\mcl{G}_{e_2}$. 
So by (F3), $e_1 = e$ and $e_2 = e$ can both be derived. 
Therefore, $e_1 = e_2$ can be derived simply by the symmetry and transitivity of equivalence $=$.  

\item[(II)]
To show that $\mcl{H}$ is an LLEE chart, 
a graph transformation was carried out from either $\mcl{G}_{e_1}$ or 
$\mcl{G}_{e_2}$ into $\mcl{H}$ step by step. 
During each step, the transformation merges two selected bisimulation nodes in $\mcl{G}_{e_1}$ (or $\mcl{G}_{e_2}$), while preserving that the graph after the merge is still an LLEE chart. 
\end{enumerate}

Our approach poses a modification to the minimization strategy on stage (II),   namely stage (II'): 
We directly prove that $\mcl{H}$ is an LEE chart according to the bisimulation function $\theta_1$ from $\mcl{G}_{e_1}$ to $\mcl{H}$ (or $\theta_2$ from $\mcl{G}_{e_2}$ to $\mcl{H}$), 
making use of the image reflection on $\mcl{G}_{e_1}$ (or $\mcl{G}_{e_2}$) by $\mcl{H}$ through $\theta_1$ (or $\theta_2$). 
Since each LEE chart is also an LLEE chart (by Prop.~\ref{prop:From LEE Witnesses to LLEE Witnesses}), 
$\mcl{H}$ is an LLEE chart.

\begin{figure}[tbhp]
    \centering
    \noindent\makebox[\textwidth]{ 
    \begin{tabular}{c}
    \begin{tikzpicture}[->,>=stealth', node distance=1.7cm]

     \coordinate (x) at (0,0);
     \coordinate (xi) at ($(x) + (0, 1.25cm)$);
     \coordinate (X) at ($(x) + (-4cm,0)$);
     \coordinate (H) at ($(X) + (0,-1.25cm)$);
     \coordinate (G) at ($(x) + (0,-1.25cm)$);
     \coordinate (i) at ($(x) + (0.5cm, -0.5cm)$);
     
    \node[circle,fill,inner sep=1.25pt, outer sep=1pt, label=south:\mbox{$\underline{x}$}] at (x) (nx) {};
    \node[circle,fill,inner sep=1.25pt, outer sep=1pt, label={[label distance = 0.1cm]east:\mbox{$\underline{x'}$}}] at (xi) (nxi) {};
    \node[circle,fill,inner sep=1.25pt, outer sep=1pt, 
    label=-105:\mbox{$\underline{X}$}, initial, initial text=] at (X) (nX) {};
    \node[draw=none] at (G) (nG) {$\hat{\mcl{G}}/\mcl{G}$};
    \node[draw=none] at (H) (nH) {$\hat{\mcl{H}}/\mcl{H}$};
    \node[draw=none] at (i) (ni) {};
    
     \path
     (nx) edge [out=180,in=180, thick, red] node[left,yshift=0.28cm] {$[2]$} node[left, yshift=0cm,black] {$a$} (nxi)
     (nx) edge [out=135,in=235, thick, red] node[right, yshift=0.25cm,xshift=0.1cm] {$[2]$} node [right, xshift=-0.1cm,black,yshift=0cm] {$b$} (nxi)
     (nxi) edge [out=-15,in=15] node[right, xshift=0cm,black] {$b$} (nx)
     (nxi) edge [out=45,in=135,loop, looseness=10, thick, red] node[right,xshift=0.1cm] {$[1]$} node[above, xshift=0cm,black] {$a$} (nxi)
     (nX) edge [out=45,in=135,loop, looseness=20, thick, red] node[right,xshift=0.2cm] {$[1]$} node[above, xshift=0cm,black] {$a$} (nX)
     (nX) edge [out=45,in=-45,loop, looseness=20, thick, red] node[right,yshift=0.2cm] {$[1]$} node[right, yshift=-0.25cm,black] {$b$} (nX)
     (nx) edge[dashed, out=210, in=-30, -, thick, gray] node[above] {$\theta$} (nX)
     (nxi) edge[dashed, out=135, in=60, -, thick, gray] node[above] {$\theta$} (nX)
     (ni) edge[] (nx)
    ;

    \coordinate (EiI) at ($(X) + (-3cm,0)$);
    \coordinate (EiII) at ($(EiI) + (0, 0.5cm)$);
    \coordinate (equH) at ($(EiI) + (0,-1.25cm)$);
     
    \node[draw=none] at (EiI) (nEiI) {
    where $s(X)\dddef (a + b)\Star 0$
    };
    \node[draw=none] at (EiII) (nEiII) {$s(X) = a\cdot s(X) + b\cdot s(X)$,};
    \node[draw=none] at (equH) (nequH) {$\mcl{H}$'s equations};

    \coordinate (EiiI) at ($(x) + (4cm,0)$);
    \coordinate (EiiII) at ($(EiiI) + (0, 1.5cm)$);
    \coordinate (equG) at ($(EiiI) + (0,-1.25cm)$);
    
    \node[draw=none] at (EiiII) (nEiiII) {
    $
    \left\{
    \begin{aligned}
        s(x) &= a\cdot s(x') + b\cdot s(x'),\\
        s(x') &= a\cdot s(x') + b\cdot s(x),
    \end{aligned}
    \right.
    $
    };

     \node[draw=none] at (EiiI) (nEiiI) {
    $
    \begin{aligned}
    \mbox{where}&\\
        s(x) &\dddef ((a+b)\cdot a\Star b)\Star 0,\\
        s(x') &\dddef a\Star b\cdot ((a+b)\cdot a\Star b)\Star 0
    \end{aligned}
    $
    };
    
    \node[draw=none] at (equG) (nequG) {$\mcl{G}$'s equations};



    \path
     ;

    \coordinate (P1) at ($(X) + (1.25cm, -0.75cm)$);
    \coordinate (P2) at ($(X) + (-0.75cm, 1cm)$);  
    \coordinate (P3) at ($(xi) + (0.75cm,-0.25cm)$);
    \coordinate (P4) at ($(xi) + (-0.5cm,0.75cm)$);

    \draw [-,dotted, thick] (P1) rectangle (P2);
    \draw [-,dotted, thick] (P3) rectangle (P4);

    \coordinate (P5) at ($(P1) + (0.15cm, 1.25cm)$);
    \coordinate (P6) at ($(P4) + (-0.15cm, -0.55cm)$);
    \coordinate (P7) at ($(P2) + (0.75cm, 0.15cm)$);
    \coordinate (P8) at ($(P4) + (-0.15cm,0.15cm)$);
    
    \draw [thick, -To, line width=1.5pt, out=180, in=0] (P6) to node[above left, xshift=0.2cm, yshift=-0.1cm] {$\theta(\{x'\})$} (P5);
    \draw [thick, -To, line width=1.5pt, dashed, bend left] (P7) to node[above left, xshift=0.2cm, yshift=-0.1cm] {reflect} (P8);

    \coordinate (Gxi) at ($(xi) + (0.25cm,1.25cm)$);
    \node[draw=none] at (Gxi) (nGxi) {$\mcl{G}_{x'}$};
    
    \end{tikzpicture}
    \\
    \end{tabular}
    } 
        \caption{A simple example illustrating our idea}
        \label{figure:A simple example illustrating our idea}
\end{figure}
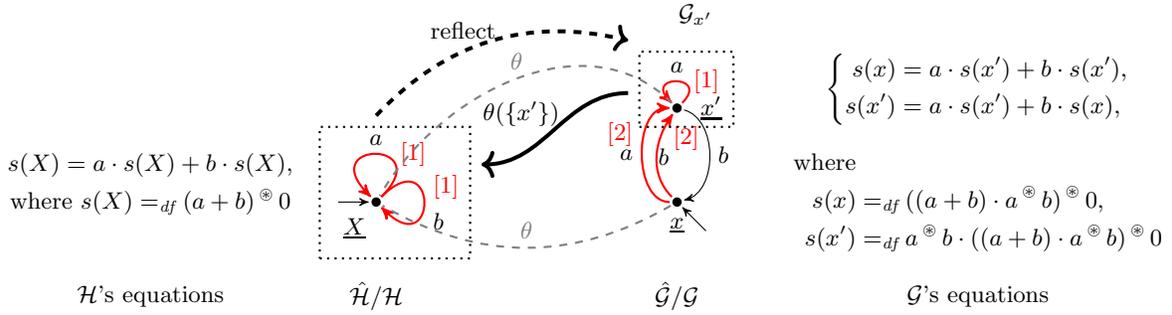

Let us see a simple example. 
In Fig.~\ref{figure:A simple example illustrating our idea}~\footnote{To distinguish node names from action names, in this paper, we underline all the node names (e.g. $\underline{X}$) in figures. }, expression $(a + b)\Star 0$ is the primary solution of 
LLEE chart $\mcl{H}$, since it is the provable solution of the equation 
shown on the left of $\mcl{H}$ that $\mcl{H}$ corresponds to.  
\ifx
$s(X) = a\cdot s(X) + b\cdot s(X)$ that $\mcl{H}$ corresponds to, where $s(X) \dddef (a + b)\Star0$. 
\fi
Expression $((a + b)\cdot a\Star b)\Star 0$ is the primary solution of LLEE chart $\mcl{G}$, which corresponds to the set of equations on its right. 
\ifx
: 
$s(x) = a\cdot s(x') + b\cdot s(x')$ and $s(x') = a\cdot s(x') + b\cdot s(x')$, where 
$s(x) \dddef ((a + b)\cdot a\Star b)\Star 0$ and $s(x') = a\Star b\cdot ((a + b)\cdot a\Star b)\Star 0$. 
\fi
$X$ and $x$ are the initial nodes of the two charts $\mcl{H}$ and $\mcl{G}$ respectively (that's why expressions $s(X)$ and $s(x)$ are called `primary' solutions). 
$\mcl{H}$ is the bisimulation collapse of $\mcl{G}$, through a bisimulation function $\theta$ from $\mcl{G}$ to $\mcl{H}$ defined as: $\theta\dddef \{x\mapsto X, x'\mapsto X\}$. 

By (F4), $(a+b)\Star 0$ is also the primary solution of $\mcl{G}$. 
To see this, let $s'(x)\dddef (a + b)\Star 0$ and $s'(x') \dddef (a + b)\Star 0$. 
From that $s(X) = a\cdot s(X) + b\cdot s(X)$ can be derived, it is easy to see that 
both $s'(x) = a\cdot s'(x') + b\cdot s'(x')$ and 
$s'(x') = a\cdot s'(x') + b\cdot s'(x)$ can be derived. 
Because $s(X)$, $s'(x)$ and $s'(x')$ are the same expression $(a + b)\Star 0$. 
Therefore, by (F3) equation $(a+b)\Star 0 = ((a + b)\cdot a\Star b)\Star 0$ can be derived. 
This forms the proof of one side on stage (I). 

Informally, to see that $\mcl{G}$ is an LLEE chart, 
by removing all tagged transitions in $\mcl{G}$ in the order indicated by the tagged number: i.e., removing transition $x'\xrightarrow{a} x'$, 
then transition $x\xrightarrow{a} x'$ and lastly transition $x\xrightarrow{b} x'$, 
all loops in $\mcl{G}$ can be eliminated. 
Similarly, it is easy to see that $\mcl{H}$ is an LLEE chart. 
The looping-back chart $\mcl{G}_{x'}$ consists of node $x'$ and transition $x'\xrightarrow{a} x'$. 
Informally, $\mcl{H}$ is the image of chart $\mcl{G}_{x'}$ through $\theta$, because $\theta$ maps to each node of $\mcl{H}$ from a node in $\mcl{G}_{x'}$: i.e., $\theta(x') = X$. 
$\mcl{G}_{x'}$ is well-structured as there is no proper sub-chart of $\mcl{G}_{x'}$ that has $\mcl{H}$ as its image.

\ifx
Informally, here, to see that $\mcl{G}$ is an LLEE chart, 
observe that $\mcl{G}$ is formed of two looping-back charts $\mcl{G}_{x'}$ and $\mcl{G}$ itself, 
where $\mcl{G}_{x'}$ is a loop sub-chart of $\mcl{G}$ consisting of node $x'$ and transition $x'\xrightarrow{a} x'$. 
These two sub-charts of $\mcl{G}$ satisfy that 
by eliminating all the tagged transitions `$\xrightarrow{}_{[\cdot]}$' of $\mcl{G}_{x'}$ and $\mcl{G}$ in 
an inside-out manner (i.e., first $\mcl{G}_{x'}$ then $\mcl{G}$ because $\mcl{G}_{x'}$ is a sub-chart of $\mcl{G}$), 
all loops in $\mcl{G}$ can be eliminated.
In this sense, 
it is easy to see that $\mcl{H}$ is also an LLEE chart (where $\mcl{H}$ itself is a looping-back chart). 
Informally, we say that $\mcl{H}$ is the image of the looping-back chart $\mcl{G}_{x'}$ through $\theta$, because for each node of $\mcl{H}$ (i.e. $X$), there is a node in $\mcl{G}_{x'}$ (i.e. $x'$) which the function $\theta$ maps to it (i.e. $\theta(x') = X$). 
$\mcl{G}_{x'}$ satisfies a well-structuredness property 
since there is no proper sub-chart of $\mcl{G}_{x'}$ that has $\mcl{H}$ as its image. 
\fi

In this example, 
in order to prove that $\mcl{H}$ is an LLEE chart, 
stage (II) merges every two bisimilar nodes by performing ``connect-through-to'' graph transformations (here a ``connect $x$ through-to $x'$'' transformation) in $\mcl{G}$ and thus transforms $\mcl{G}$ into $\mcl{H}$. 
\cite{Grabmayer20} generally proved that each transformation preserves the LLEE structure. 
In our opinion, there are two technical complexities in the proof work on stage (II). 
One is to choose the right bisimilar pair each time, the other is 
to analyze different cases of graph transformations based on different types of bisimilar pairs. 

Our approach (II'), on the other hand,  tackles these two technical complexities by directly showing that $\mcl{H}$, as the image of the looping-back chart $\mcl{G}_{x'}$ through $\theta$, satisfies the LEE property. 
Intuitively, the behaviour of $\mcl{H}$ is contrained by the behaviour of $\mcl{G}_{x'}$ through $\theta$ so that all loops in $\mcl{H}$ pass through $X$ and all paths starting from $X$ stay in $\mcl{H}$ before returning to $X$. 
This so-called \emph{LEE sub-chart} on $\mcl{H}$ enforced by its corresponding well-structured looping-back chart $\mcl{G}_{x'}$ is called that 
$\mcl{H}$ \emph{reflects} the LLEE structure of $\mcl{G}_x$. 

\ifx
In this example, 
in order to prove that $\mcl{H}$ is an LLEE chart, 
stage (II) merges the two bisimilar nodes $x$ and $x'$ in $\mcl{G}$ and thus transforms $\mcl{G}$ into $\mcl{H}$. 
\cite{Grabmayer20} generally proved that this transformation process preserves the LLEE structure. 
Our approach (II'), on the other hand, will show that $\mcl{H}$, as the image of the looping-back chart $\mcl{G}_{x'}$ through $\theta$, satisfies the LEE property
because of the LLEE structure of $\mcl{G}_{x'}$ and the bisimulation relation between $X$ and $x'$ indicated by $\theta$. 
We say that the image $\mcl{H}$ \emph{reflects} the LLEE structure of $\mcl{G}_{x'}$, in the sense that intuitively, 
by that all loops in $\mcl{G}_{x'}$ pass through $x'$ and all paths in $\mcl{G}_{x'}$ from $x'$ stay in $\mcl{G}_{x'}$ before returning to $x'$, 
we obtained that
all loops in $\mcl{H}$ must pass through $X$ and all paths from $X$ must stay in $\mcl{H}$ before returning to $X$. 
\fi

In this paper, 
by utilizing this phenomenon between an image and its correspondence well-structured looping-back chart, 
we will prove in general case, 
that the bisimulation collapse of any LLEE chart itself is an LLEE chart. 

\ifx
In this paper, 
by utilizing  the LLEE structure of $\mcl{G}$ and 
the bisimulation relation $\theta: \mcl{G}\to \mcl{H}$ as shown in this example,  
we will prove in general case, that the bisimulation collapse of any LLEE chart itself is an LLEE chart.  
\fi

As a summery, the main contributions of this paper are mainly twofolds:
\begin{enumerate}[$\bullet$]
    \item 
    We introduce the concepts of \emph{images} and \emph{well-structured looping-back charts} on LLEE charts, and analyze their properties and relationship. 
    \item Based on these concepts, we propose a new method for proving the completeness of \MilIfree\ by directly showing that the
    bisimulation collapse of every LLEE chart is an LLEE chart, without performing graph transformations. 
    
\end{enumerate}

The rest of the paper is organized as follows:
Sect.~\ref{section:Preliminaries} introduces the necessary background knowledge needed to understand our work. 
In Sect.~\ref{section:Looping-back Charts and well-structured Images} we define 
the concepts as the foundation of our proof method, 
while in Sect.~\ref{section:Restoration for LLEE Charts} we mainly focus on our proof. 
Sect.~\ref{section:Conclusion and Future Work} makes a conclusion and an analysis. Sect.~\ref{section:Future Work} discusses about future work.

\section{Preliminaries}
\label{section:Preliminaries}

In this section we introduce the necessary background related to this topic. 
Note that some of them are not directly used in our proof method, but are necessary for understanding the minimization strategy. We tag them as `(*)'.

\subsection{1-Free Regular Expressions and Charts}

\paragraph{1-Free Regular Expression}
A 1-free regular expression $e\in \IFSExp$ is defined as follows in BNF form:
$$
e \dddef a\ |\ 0 \ |\ e + e\ |\ e\cdot e\ |\ e\Star e,
$$
where $a\in \Act$ is called an \emph{action}, ranged over by $a,b,c,d$. 
$e\Star e$ is the binary star operator, $e_1\cdot e_2$ is sometimes abbreviated as $e_1e_2$. 
We use $\equiv$ to represent the syntactic identity between two 1-free regular expressions.

\paragraph{Chart}
A \emph{chart} is a 5-tuple $\mcl{G} = \la V, A, v, \tran, \surd\ra$, where $V$ is a finite set of \emph{nodes}, usually ranged over by $x, y, z, X, Y, Z$; $A\subseteq \Act$ is a set of actions;
$v\in V$ is the \emph{initial node};
$\tran : V\times A\times (V\cup \{\surd\})$ is a set of transitions; 
$\surd$ is a special node called \emph{termination}. 
 Each transition $\la X, a, Y\ra\in \tran$ and transition $\la X, a, \surd\ra\in \tran$ are also written as $X\xrightarrow{a} Y$ and $X\xrightarrow{a}\surd$ respectively, where $a\in A$. 
 We also call $X\xrightarrow{a}\surd$ a \emph{terminal transition}. 
A sequence of transitions is often written as: $X_1\xrightarrow{a_1} X_2\xrightarrow{a_2}...\xrightarrow{a_n}X_n\xrightarrow{b} \xi$ ($n > 0$), where 
$\xi$ is a node or $\surd$, 
$X_1\xrightarrow{a_1} X_2$, ..., $X_n\xrightarrow{b} \xi$ are $n$ transitions. 
In this paper, any chart we talk about satisfies that all nodes are reachable from the initial node. 

\ifx
For convenience, when the structure of a chart $\mcl{G}$ in the context is clear or irrelevant, we often mix the concepts of the set of nodes of $\mcl{G}$ and the chart $\mcl{G}$ itself consisting of its set of nodes and all transitions between them. And we simply use $\mcl{G}$ to represent both concepts. 
For example, given two charts $\mcl{G} = \la V_G, A_G, s_G, \tran_G, \surd\ra$ and $\mcl{H} = \la V_H, A_H, s_H, \tran_H, \surd\ra$, 
we often write $x\in \mcl{G}$ to mean $x\in V_G$; 
write $\mcl{G}\subseteq \mcl{H}$ to mean both $V_G\subseteq V_H$ and 
that $\mcl{G}$ is a sub-chart of $\mcl{H}$; 
we will use $\mcl{G}\cup \mcl{H}$ to express both the set $V_G\cup V_H$ and the chart made of the set $V_G\cup V_H$ and all transitions between these nodes. 
\fi

\ifx
For convenience, when no ambiguity is caused, we often directly use the chart symbol, like $\mcl{G}, \mcl{H}$ and $\mcl{C}$, to express its set of nodes. 
In a given chart, we often mix using the concept of a set of nodes and 
the concept of the sub-chart consisting of these nodes and all transitions between them. 
e.g., 
for the sub-chart $\mcl{G}_{x'}: \la \{x'\}, \{a\}, x', \{x'\xrightarrow{a} x'\}, \surd\ra$ of $\mcl{G}$ in Fig.~\ref{figure:A simple example illustrating our idea}, 
we can write $x'\in \mcl{G}_{x'}$ to mean that $x'$ is a node of $\mcl{G}_{x'}$;
$\mcl{G}_{x'}$ can also be simply denoted as its set of nodes: $\{x'\}$. 
Given two sub-charts $\mcl{G} = \la V_G, A_G, s_G, \tran_G, \surd\ra$ and $\mcl{H} = \la V_H, A_H, s_H, \tran_H, \surd\ra$, 
depending on the context, 
we will use 
$\mcl{G}\cup \mcl{H}$ to express either the set $V_G\cup V_H$ or the chart 
$\la V_G\cup V_H, A_G\cup A_H, v_{G\cup H}, \tran_G\cup \tran_H, \surd \ra$ (where we ignore the discussion of its initial node $v_{G\cup H}$). 
\fi

Sometimes we ignore symbol $a$ and simply write a transition $X\xrightarrow{a} Y$ as $X\xrightarrow{}Y$. 
We use $X\eqvn Y$ to represent that nodes $X$ and $Y$ are identical. 
Use $\cdot$ to represent an arbitrary node, e.g. a transition $X\xrightarrow{} \cdot$. 
The reflexive and transitive closure of $\xrightarrow{}$ is defined as:
$X\xrightarrow{}^* Y$ if either (i) $X\eqvn Y$ or (ii) $X\xrightarrow{}Z$
and $Z\xrightarrow{}^* Y$ for some $Z$. 
The transitive closure $X\xrightarrow{}^+ Y$ is defined s.t. $X\xrightarrow{}Z$ and $Z\xrightarrow{}^* Y$ for some $Z$.
A \emph{path} is a finite or infinite sequence of transitions. 
We use $X\xrightarrow[\not \eqvn Z]{} Y$, $X\xrightarrow[\not\eqvn Z]{}^* Y$ or $X\xrightarrow[\not\eqvn Z]{}^+ Y$ to represent a path where all intermediate nodes (not including $X$) are not $Z$.
A \emph{loop} is a finite path of the form: $X\xrightarrow{}^+ X$, i.e., it starts from and returns to the same node $X$. 

\paragraph{Chart Interpretation (*)}
The \emph{chart interpretation} of a 1-free regular expression $e\in \IFSExp$ is a chart 
$$
\mcl{C}(e) \dddef \la V(e), A(e), e, \tran(e), \surd\ra, 
$$
which is obtained by expanding $e$ according to its operational semantics given in Table~\ref{table:Operational Semantics of 1-free Regular Expressions}. 
$e$ is the initial node. 
$\tran(e)$ consists of all the transitions during the expansion of $e$. 
$V(e)$ is the finite set of expressions reachable from $e$ through transitions in $\tran(e)$. 
$A(e)$ is the set of actions appeared in $e$. 

\begin{table}[tbph]
	\centering
	\begin{tabular}{c}
		\toprule
        \\
		\infer[]
        {a\xrightarrow{a} \surd}
        {}
        \ \ \ 
        \infer[]
        {e_1 + e_2\xrightarrow{a} \surd}
        {e_1\xrightarrow{a} \surd}
        \ \ \ 
        \infer[]
        {e_1 + e_2\xrightarrow{a} \surd}
        {e_2\xrightarrow{a} \surd}
        \ \ \ 
        \infer[]
        {e_1\cdot e_2\xrightarrow{a} e_2}
        {
        e_1\xrightarrow{a}\surd
        }
        \ \ \ 
        \infer[]
        {
        e_1\Star e_2\xrightarrow{a} e_1\Star e_2
        }
        {
        e_1\xrightarrow{a} \surd
        }
        \ \ \ 
        \infer[]
        {
        e_1\Star e_2\xrightarrow{a}\surd
        }
        {e_2\xrightarrow{a} \surd}
        \\
        \\
        \infer[]
        {e_1 + e_2\xrightarrow{a} e'_1}
        {e_1\xrightarrow{a} e'_1}
        \ \ \ 
        \infer[]
        {e_1 + e_2\xrightarrow{a} e'_2}
        {e_2\xrightarrow{a} e'_2}
        \ \ \ 
        \infer[]
        {e_1 \cdot e_2\xrightarrow{a} e'_1 \cdot e_2}
        {e_1\xrightarrow{a} e'_1}
        \\
        \\ 
        \infer[]
        {e_1\Star e_2\xrightarrow{a} e'_1\cdot (e_1\Star e_2)}
        {e_1\xrightarrow{a} e'_1}
        \ \ \ 
        \infer[]
        {e_1\Star e_2\xrightarrow{a} e'_2}
        {e_2\xrightarrow{a} e'_2}
        \\
        \\
		\bottomrule
	\end{tabular}
	\caption{Operational semantics of 1-free regular expressions}
	\label{table:Operational Semantics of 1-free Regular Expressions}
\end{table}

\subsection{Bisimulation Equivalence and Proof System \MilIfree}

\paragraph{Bisimulation Equivalence}

A \emph{bisimulation relation} $\R\subseteq \mbb{V}\times \mbb{V}$, where $\mbb{V}$ is the universal set of nodes (i.e., the set of all nodes in our discussed domain), is a binary relation satisfying that for any $(X, Y)\in \R$, the following hold:
\begin{enumerate}[(1)]
	\item if $X\xrightarrow{a}X'$, then there is some $Y''$ such that $Y\xrightarrow{a}Y''$ and $(X', Y'')\in \R$;
	\item if $Y\xrightarrow{a}Y'$, then there is some $X''$ such that $X\xrightarrow{a}X''$ and $(X'', Y')\in \R$;
	\item $X \xrightarrow{a}\surd $ iff $Y \xrightarrow{a} \surd$. 
\end{enumerate}

\ifx
Given two charts $\mcl{G}=\la V_1, A_1, s_1, \tran_1, \surd\ra$ and 
$\mcl{H}=\la V_2, A_2, s_2, \tran_2, \surd\ra$, 
a \emph{bisimulation relation} $\R\subseteq V_1\times V_1$ 
between two charts $\mcl{G}$ and $\mcl{H}$
is a binary relation satisfying that
\begin{enumerate}[(1)]
        \item $(s_1, s_2)\in \R$;
	\item if $X\xrightarrow{a}X'\in \tran_1$, then $Y\xrightarrow{a}Y'\in \tran_2$ and $(X', Y')\in \R$;
	\item if $Y\xrightarrow{a}Y'\in \tran_2$, then $X\xrightarrow{a}X'\in \tran_1$ and $(X', Y')\in \R$;
	\item $X \xrightarrow{a}\surd \in \tran_1$ iff $Y \xrightarrow{a} \surd\in \tran_2$. 
\end{enumerate}
\fi

Given two charts $\mcl{G}=\la V_1, A_1, v_1, \tran_1, \surd\ra$ and 
$\mcl{H}=\la V_2, A_2, v_2, \tran_2, \surd\ra$, 
a \emph{bisimulation relation} $\R\subseteq V_1\times V_2$ \emph{between} $\mcl{G}$ \emph{and} $\mcl{H}$ satisfies that $(v_1, v_2)\in \R$. 
We also denote $\R : V_1\times V_2$ as $\R: V_1\to V_2$, and 
call $V_1$ the \emph{domain} and $V_2$ the \emph{codomain} of relation $\R$. 
Given $\R: V_1\to V_2$, 
for any set $A\subseteq V_1$ and $B\subseteq V_2$, we define 
$\R(A)\dddef \{X\ |\ X\in V_2, \exists y\in A. (y, X)\in \R\}$;
and define 
$\R^{-1}(B)\dddef \{x\ |\ x\in V_1, \exists Y\in B. (x, Y)\in \R\}$. 
A bisimulation relation $\R: V_1\to V_2$ is called a \emph{bisimulation function} if $\R$ is a function from $V_1$ to $V_2$.  
We often use $\theta$ to represent a bisimulation function. 


Two nodes $X$ and $Y$ are \emph{bisimilar}, denoted by $X\sim Y$, if there is a bisimulation relation $\R$ such that $(X, Y)\in \R$. 
Two charts $\mcl{G}$ and $\mcl{H}$ are said bisimilar, denoted by $\mcl{G} \sim \mcl{H}$, if there is a bisimulation relation $\R$ between $\mcl{G}$ and $\mcl{H}$. 
The \emph{bisimulation equivalence} $e_1 = e_2$ between two 1-free regular expressions $e_1, e_2\in \IFSExp$, is defined 
if $\mcl{C}(e_1)\sim \mcl{C}(e_2)$.

\paragraph{Proof System \MilIfree\ (*)}
Proof system \MilIfree\ for 1-free regular expressions modulo bisimulation equivalence $=$ is given in Table~\ref{table:1-free proof system for 1-free regular expressions} as a set of rules on equations, where (A1) - (A9) are axioms, (R1) is an inference rule. 
Since our work in this paper will not concern this proof system, 
we will not explain the meaning of each rule. One can refer to~\cite{Grabmayer20} for more details. 

\begin{table}[tbph]
	\centering
	\begin{tabular}{c}
		\toprule
        $
        {e_1 + e_2 = e_2 + e_1}\ (A1)
        $
	\ \ \ 
        $
	{(e_1 + e_2) + e_3 = e_1 + (e_2 + e_3)}\ (A2)
        $	
    \ \ \ 
    $
    {e + e = e}\ (A3)
    $
    \\
    \\
    $
    {(e_1 + e_2)\cdot e_3 = e_1\cdot e_3 + e_2\cdot e_3}\ (A4)
    $
	\ \ \ 
    $
    {(e_1\cdot e_2)\cdot e_3 = e_1\cdot (e_2\cdot e_3)}\ (A5)
    $	
  \ \ \ 
  $
  {e + 0 = e}\ (A6)
  $
  \\
  \\
  $
  {0\cdot e = 0}\ (A7)
  $
  \ \ \ 
  $
  {e_1\Star e_2 = e_1\cdot (e_1\Star e_2) + e_2}\ (A8)
  $
  \ \ \ 
  $
  {(e_1\Star e_2)\cdot e_3 = e_1\Star (e_2\cdot e_3)}\ (A9)
  $
   \\
    \\
    $
    \infer[(R1)]
    {
        e = e_1\Star e_2
    }
    {
        e = e_1\cdot e + e_2
    }
    $
    \\
		\bottomrule
	\end{tabular}
	\caption{Proof system \MilIfree\ for 1-free regular expressions}
	\label{table:1-free proof system for 1-free regular expressions}
\end{table}

\ifx
\begin{table}[tbph]
	\centering
	\begin{tabular}{c}
		\toprule
        $
        \infer[(A1)]
        {e_1 + e_2 = e_2 + e_1}
        {}
        $
	\ \ \ 
        $\infer[(A2)]
	   {(e_1 + e_2) + e_3 = e_1 + (e_2 + e_3)}{}$	
    \ \ \ 
    $
    \infer[(A3)]
	   {e + e = e}{}
    $
    \\
    \\
    $\infer[(A4)]
    {(e_1 + e_2)\cdot e_3 = e_1\cdot e_3 + e_2\cdot e_3}{}$
	\ \ \ 
    $\infer[(A5)]
    {(e_1\cdot e_2)\cdot e_3 = e_1\cdot (e_2\cdot e_3)}{}$	
  \ \ \ 
  $\infer[(A6)]
  {e + 0 = e}
  {}
  $
  \\
  \\
  $
  \infer[(A7)]
  {0\cdot e = 0}{}
  $
  \ \ \ 
  $
  \infer[(A8)]
  {e_1\Star e_2 = e_1\cdot (e_1\Star e_2) + e_2}{}
  $
  \ \ \ 
  $
  \infer[(A9)]
  {(e_1\Star e_2)\cdot e_3 = e_1\Star (e_2\cdot e_3)}{}
  $
   \\
    \\
    $
    \infer[(R1)]
    {
        e = e^*_1\cdot e_2
    }
    {
        e = e_1\cdot e + e_2
    }
    $
    \ \ \ 
    $
    \infer[(R2)]
    {e[e_1] = e[e_2]}
    {e_1 = e_2}
    $
    \\
		\bottomrule
	\end{tabular}
	\caption{Proof system \MilIfree\ for 1-free regular expressions}
	\label{table:1-free proof system for 1-free regular expressions}
\end{table}
\fi

Given an equation $e_1 = e_2$, we write $\IFvdash e_1 = e_2$ to mean that $e_1 = e_2$ can be derived by applying the rules in Table~\ref{table:1-free proof system for 1-free regular expressions} in a finite number of steps. 
Proof system \MilIfree\ is sound, that is,  $\IFvdash e_1 = e_2$ implies $e_1 = e_2$ for any $e_1, e_2\in \IFSExp$. 
The completeness of \MilIfree\ was first obtained as the main result in~\cite{Grabmayer20}, which means that for every semantically-correct equation $e_1 = e_2$ with $e_1, e_2\in \IFSExp$, $\IFvdash e_1 = e_2$. 


\paragraph{Chart and 1-Free Regular Expression Equations (*)}
An \emph{evaluation} $s: \mbb{V}\to \IFSExp$ defines a function that maps each node to a 1-free regular expression. 
Given an evaluation $s$, a chart $\mcl{G}=\la \{X_i\}_{i\in [1,n]}, A, X_1, \tran, \surd\}\ra$ corresponds to 
a set of equations w.r.t. $s$:
$$s(\mcl{G})\dddef \left\{s(X_i) = \sum^n_{j=1}e_{ij}\cdot s(X_j) + \sum^m_{k=1}f_{ik} \right\}_{i\in [1,n]},$$ where in each equation $s(X_i) = \sum^n_{j=1}e_{ij}\cdot s(X_j) + \sum^m_{k=1}f_{ik}$, term $e_{ij}\cdot s(X_j)$ ($e_{ij}\in A$) appears iff $X_i\xrightarrow{e_{ij}}X_j$ is a transition in $\tran$; 
term $f_{ik}$ appears iff $X_i\xrightarrow{f_{ik}} \surd$ is a transition in $\tran$. 
\ifx
We sometimes call $s(X_i)$ a \emph{node expression} and call $s(X_i) = \sum^n_{j=1}e_{ij}\cdot s(X_j) + \sum^m_{k=1}f_{ik}$ the \emph{equation of node $X_i$ (w.r.t. $s$)}. 
\fi
An evaluation $s$ is a \emph{provable solution} of $\mcl{G}$, if each equation of $s(\mcl{G})$ can be derived according to proof system \MilIfree\ in Table~\ref{table:1-free proof system for 1-free regular expressions}, denoted by $\IFvdash s(\mcl{G})$. 
Call $s(X_1)$ a \emph{primary solution} of $\mcl{G}$.

\subsection{Loop Charts, LEE Charts and LLEE Charts~\cite{Grabmayer20}}
\label{section:LEE-LLEE Charts}

\paragraph{Loop Chart}
A \emph{loop chart} $\mcl{A} = \la V, A, X, \tran, \surd\ra$ is a chart with the initial node $X$ called \emph{start node} and satisfying the following three conditions:
\begin{enumerate}
	\item[(L1)] There is an infinite path from $X$ in $\mcl{A}$;
	\item[(L2)] In $\mcl{A}$ every infinite path from $X$ returns to $X$ after a finite number of transitions;
	\item[(L3)] $\surd$ is not in $\mcl{A}$. 
\end{enumerate}
The set of transitions starting from $X$ is called \emph{loop-entry transitions}. 
Other transitions are called \emph{body transitions}. 
Call $V_\mcl{A} \setminus \{X\}$ the \emph{body} of $\mcl{A}$. 

For example,
in Fig.~\ref{figure:A simple example illustrating our idea}, chart $\mcl{H}$ is a loop chart. Chart $\mcl{G}$ with the start node $x$ is not a loop chart as it violates (L2). 

In a chart $\mcl{G}=\la V, A, v, \tran, \surd\ra$, 
a \emph{$\la X, \E\ra$-generated chart} is a sub-chart $\mcl{C} = \la V_1, A, X, \tran_1, \surd\ra$ of $\mcl{G}$ consisting of all the transitions each of which 
is on a path that starts from a transition $X\xrightarrow{} Y$ in $\E$, 
and continues with other transitions of $\mcl{G}$ until reaching $X$ again. 
Formally, 
$\tran_1 \dddef \E\cup \{Y\xrightarrow{} Z\ |\ X\xrightarrow[\not\eqvn X]{} X_1\xrightarrow[\not\eqvn X]{}^* Y\xrightarrow{} Z, (X\xrightarrow{} X_1)\in \E\}\cup \{Y\xrightarrow{} \surd\ |\ X\xrightarrow[\not\eqvn X]{} X_1\xrightarrow[\not\eqvn X]{}^* Y\xrightarrow{} \surd, (X\xrightarrow{} X_1)\in \E\}$, $V_1$ is the set of nodes appearing in $\tran_1$. 

A $\la X, \E\ra$-generated chart of $\mcl{G}$ is called a \emph{loop sub-chart} of $\mcl{G}$ if 
it is a loop chart. 

In Fig.~\ref{figure:A simple example illustrating our idea}, chart $\mcl{H}$,
as the $\la X, \{X\xrightarrow{a} X, X\xrightarrow{b} X\}\ra$-generated chart,
is a loop sub-chart of itself. 
Chart $\mcl{G}_{x'}: \la \{x'\}, \{a, b\}, x', \{x'\xrightarrow{a}x'\}, \surd\ra$, as the $\la x', \{x'\xrightarrow{a} x'\}\ra$-generated chart, is a loop sub-chart of $\mcl{G}$.

\paragraph{LEE Chart}
An \emph{elimination of a loop chart} $\mclul{A}$ startinig from $X$ is a transformation process in which we first remove all of its loop-entry transitions from $X$, then remove all the nodes and transitions that become unreachable.  
A chart $\mclul{G}$ is said to satisfy the \emph{loop existence and elimination} (LEE) property, if there exists a transformation process in which each loop sub-chart of $\mcl{G}$ can be eliminated step by step, so that after 
the process, $\mclul{G}$ results in a chart without an infinite path. 
A chart is called an \emph{LEE chart} if it satisfies the LEE property. 

\begin{figure}[tbhp]
    \centering
    \noindent\makebox[\textwidth]{ 
    \begin{tabular}{c}
    \begin{tikzpicture}[->,>=stealth', node distance=1.7cm]

     \coordinate (x) at (0,0);
     \coordinate (xi) at ($(x) + (0, 1.25cm)$);
     \coordinate (G) at ($(x) + (0,-1cm)$);
     
    \node[circle,fill,inner sep=1.25pt, outer sep=1pt, label=south:\mbox{$\underline{x}$}] at (x) (nx) {};
    \node[circle,fill,inner sep=1.25pt, outer sep=1pt, label={[label distance = 0.1cm]east:\mbox{$\underline{x'}$}}] at (xi) (nxi) {};
    \node[draw=none] at (G) (nG) {$\mcl{G}$};
    
     \path
     (nx) edge [out=180,in=180] node[left,yshift=0.28cm] {} node[left, yshift=0cm,black] {$a$} (nxi)
     (nx) edge [out=135,in=235] node[right, yshift=0.25cm,xshift=0.1cm] {} node [right, xshift=-0.1cm,black,yshift=0cm] {$b$} (nxi)
     (nxi) edge [out=-15,in=15] node[right, xshift=0cm,black] {$b$} (nx)
     (nxi) edge [out=45,in=135,loop, looseness=10, thick, red] node[right,xshift=0.1cm] {} node[above, xshift=0cm,black] {$a$} (nxi)
    ;

    \coordinate (Ix) at ($(x) + (3.5cm,0)$);
     \coordinate (Ixi) at ($(Ix) + (0, 1.25cm)$);
     \coordinate (IG) at ($(Ix) + (0,-1cm)$);
     
    \node[circle,fill,inner sep=1.25pt, outer sep=1pt, label=south:\mbox{$\underline{x}$}] at (Ix) (nIx) {};
    \node[circle,fill,inner sep=1.25pt, outer sep=1pt, label={[label distance = 0.1cm]east:\mbox{$\underline{x'}$}}] at (Ixi) (nIxi) {};
    \node[draw=none] at (IG) (nIG) {$\mcl{G}'$};
    
     \path
     (nIx) edge [out=180,in=180, thick, red] node[left,yshift=0.28cm] {} node[left, yshift=0cm,black] {$a$} (nIxi)
     (nIx) edge [out=135,in=235, thick, red] node[right, yshift=0.25cm,xshift=0.1cm] {} node [right, xshift=-0.1cm,black,yshift=0cm] {$b$} (nIxi)
     (nIxi) edge [out=-15,in=15] node[right, xshift=0cm,black] {$b$} (nIx)
    ;

    \coordinate (IIx) at ($(Ix) + (3.5cm,0)$);
    \coordinate (IIG) at ($(IIx) + (0,-1cm)$);
     
    \node[circle,fill,inner sep=1.25pt, outer sep=1pt, label=south:\mbox{$\underline{x}$}] at (IIx) (nIIx) {};
    \node[draw=none] at (IIG) (nIIG) {$\mcl{G}''$};
    
     \path
    ;

    \coordinate (P1) at ($(x) + (0.5cm,-0.25cm)$);
    \coordinate (P2) at ($(Ix) + (-0.5cm,-0.25cm)$);
    \coordinate (P3) at ($(Ix) + (0.5cm,-0.25cm)$);
    \coordinate (P4) at ($(IIx) + (-0.5cm,-0.25cm)$);

    \draw [thick, -To, line width=1.5pt] (P1) to node[above] {transform} (P2);
    \draw [thick, -To, line width=1.5pt] (P3) to node[above] {transform} (P4);

    \coordinate (Cx) at ($(x) + (-3cm,0)$);
     \coordinate (Cxi) at ($(Cx) + (0, 1.25cm)$);
     \coordinate (CG) at ($(Cx) + (0,-1cm)$);
     
    \node[circle,fill,inner sep=1.25pt, outer sep=1pt, label=south:\mbox{$\underline{x}$}] at (Cx) (nCx) {};
    \node[circle,fill,inner sep=1.25pt, outer sep=1pt, label={[label distance = 0.1cm]east:\mbox{$\underline{x'}$}}] at (Cxi) (nCxi) {};
    \node[draw=none] at (CG) (nCG) {$\hat{\mcl{G}}/\mcl{G}$};
    
     \path
     (nCx) edge [out=180,in=180,thick,red] node[left,yshift=0.28cm] {$[2]$} node[left, yshift=0cm,black] {$a$} (nCxi)
     (nCx) edge [out=135,in=235, thick,red] node[right, yshift=0.25cm,xshift=0.1cm] {$[2]$} node [right, xshift=-0.1cm,black,yshift=0cm] {$b$} (nCxi)
     (nCxi) edge [out=-15,in=15] node[right, xshift=0cm,black] {$b$} (nCx)
     (nCxi) edge [out=45,in=135,loop, looseness=10, thick, red] node[right,xshift=0.1cm] {$[1]$} node[above, xshift=0cm,black] {$a$} (nCxi)
    ;

    \end{tikzpicture}
    \\
    \end{tabular}
    } 
        \caption{An example of elimination processes}
        \label{figure:An example of elimination processes}
\end{figure}
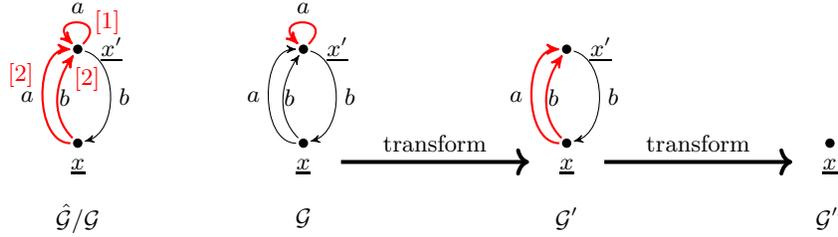

Fig.~\ref{figure:An example of elimination processes} shows how the LEE chart $\mcl{G}$ of Fig.~\ref{figure:A simple example illustrating our idea} can be eliminated in two steps: 
firstly the loop sub-chart $\mcl{G}_{x'}$ is eliminated by removing the loop-entry transition $x'\xrightarrow{a} x'$; after that the resulted chart $\mcl{G}'$, which becomes a loop sub-chart of itself, is then eliminated by removing two loop-entry transitions $x\xrightarrow{a} x'$ and $x\xrightarrow{b} x'$, and then removing node $x'$ and transition $x'\xrightarrow{b} x$ that become unreachable.
This leaves the resulted chart $\mcl{G}''$ having no infinite paths.

As illustrated in the example above, 
the process of loop elimination is actually conducted in an inside-out manner, in the sense that 
a sub-chart $\mcl{A}$ of $\mclul{G}$ can only become a loop sub-chart after all the loop sub-charts of $\mcl{A}$ have been eliminated previously. 
This is an important intuition of our method purposed in Sect.~\ref{section:Restoration for LLEE Charts}. 

\ifx
Given an elimination process on $\mcl{G}$, 
we call a sub-chart $\mcl{A}$ of $\mcl{G}$ a \emph{loop sub-chart} (w.r.t. this process), if $\mcl{A}$ is or can become a loop chart during this process.  
We use $\mcl{A}(X, \E_X)$ (or simply $\mcl{A}_X$ by omitting $\E_X$) to represent 
a loop sub-chart of $\mcl{G}$, with $X$ its start node and $\E_X$ the set of its loop-entry transitions starting from $X$. 
\fi

For an elimination process on an LEE chart $\mcl{G}$, a function $\hat{\mcl{G}}$, called an \emph{LEE witness}, is used to indicate the order of removing the loop-entry transitions in $\mcl{G}$. 
$\hat{\mcl{G}}$ maps the transitions of $\mcl{G}$ to a down-closed set of natural numbers, we call \emph{order numbers}. 
It maps a loop-entry transition $X\xrightarrow{} Y$ to a number $n > 0$, denoted by $X\xrightarrow{}_{[n]} Y$, if $X\xrightarrow{} Y$ is removed at the $n$th step of the process; and maps body transitions to $0$, denoted by $\cdot \xrightarrow{}_{[0]} \cdot$ or $\cdot \rightarrow_{\bo} \cdot$. 
Write an LEE chart $\mcl{G}$ as $\hat{\mcl{G}}/\mcl{G}$. 

In chart $\mcl{G}$ of Fig.~\ref{figure:An example of elimination processes}, 
the tagged numbers on arrows indicate an LEE witness $\hat{\mcl{G}}$ of the elimination process in Fig.~\ref{figure:An example of elimination processes}, 
where $\hat{\mcl{G}}(x'\xrightarrow{a} x') \dddef 1$, $\hat{\mcl{G}}(x\xrightarrow{a} x') \dddef 2$, $\hat{\mcl{G}}(x\xrightarrow{b} x') \dddef 2$, $\hat{\mcl{G}}(x'\xrightarrow{b} x)\dddef 0$.

\paragraph{LLEE Chart}
The elimination process of an LEE chart is not necessarily unique. 
If in an LEE chart there exists an elimination process in which 
no loop-entry transitions are removed from a node in the body of a previously eliminated loop sub-chart, then say the LEE chart satisfies the \emph{layered loop existence and elimination} (LLEE) property. The LEE chart is then called an \emph{LLEE chart}. 
A witness $\hat{\mcl{G}}$ of an LLEE chart $\mcl{G}$ is thus called an \emph{LLEE witness}. 

It is easy to see that chart $\mcl{G}$ of Fig.~\ref{figure:An example of elimination processes} is also an LLEE chart.

\ifx
In an LLEE chart, a relation $X\pincycle Y$, called \emph{in-loop} relation in this paper, is defined if 
there is a path starting with a loop-entry transition from $X$ to $Y$, that is, $X\xrightarrow{}_{[m]}\cdot \xrightarrow{}^*_{\textit{bo}} Y$. 
$X\incycle Y$ is defined if $X\equiv Y$ or $X\pincycle Y$. 

\begin{prop}[\cite{???}]
    Relation $\incycle$ in an LLEE chart $\mcl{G}$ is a partially-ordered relation. 
\end{prop}
\fi

Not all LEE witnesses are LLEE witnesses. 
Considering the chart $\hat{\CI}/\CI$ in Fig.~\ref{figure:An example of a non-LLEE LEE witness}, 
from its elimination process we see that $\hat{\CI}$ is an LEE witness. 
However, it is not an LLEE witness,
because when performing its elimination at step 3, loop-entry transition $X\xrightarrow{b_1}_{[3]} Z$ is from the body of the loop sub-chart $\la \{Z, X\}, \{a_1, a_2, a_3, a_4, b_1, d_1, d_2\}, Z, \{Z\xrightarrow{a_2} X, X\xrightarrow{b_1} Z\}, \surd\ra$ of $\CI'$ after step 1.  

\begin{figure}[tbhp]
    \centering
    \noindent\makebox[\textwidth]{ 
    \begin{tabular}{c}
    \begin{tikzpicture}[->,>=stealth', node distance=1.7cm]

    \coordinate (Z) at (0,0);
    \coordinate (X) at ($(Z) + (-1.25cm,1.25cm)$);
    \coordinate (Y) at ($(Z) + (1.25cm,1.25cm)$);
    \coordinate (K) at ($(Z) + (0,1.25cm)$);
    \coordinate (CI) at ($(Z) + (0,-1.5cm)$);

    \node[circle,fill,inner sep=1.25pt, outer sep=1pt, label = {[label distance=0.15cm]-135:$\underline{Z}$}] at (Z) (nZ) {};
    \node[circle,fill,inner sep=1.25pt, outer sep=1pt, label = north:$\underline{X}$] at (X) (nX) {};
    \node[circle,fill,inner sep=1.25pt, outer sep=1pt, label = north:$\underline{Y}$] at (Y) (nY) {}; 
    \node[circle,fill,inner sep=1.25pt, outer sep=1pt, label = north:$\underline{K}$] at (K) (nK) {}; 
    \node[draw=none] at (CI) (nCI) {$\hat{\CI}/\CI$}; 

    \path
     (nZ) edge [out=180, in=-90, thick, red] node[below, yshift=-0.1cm] {$[2]$} node[left, yshift=0cm,black] {$a_2$} (nX)
     (nX) edge [thick, red] node[above, yshift=0.1cm] {$[3]$} node[right, yshift=0cm,black] {$b_1$} (nZ)
     (nZ) edge [out=0, in=-90, thick, red] node[yshift=-0.35cm] {$[1]$} node[right, yshift=0cm,black] {$a_3$} (nY)
     (nY) edge [] node[right, yshift=0cm,black] {$d_1$} (nZ)
     (nZ) edge [out=-45,in=-135,loop, thick, red] node[right,xshift=0.1cm] {$[1]$} node[below, black] {$a_1$} (nZ)
     (nZ) edge [] node[right] {$a_4$} (nK)
     (nK) edge [] node[above] {$d_2$} (nX)
     ;

    \coordinate (IZ) at ($(Z) + (4cm,0)$);
    \coordinate (IX) at ($(IZ) + (-1.25cm,1.25cm)$);
    \coordinate (IK) at ($(IZ) + (0,1.25cm)$);
    \coordinate (ICI) at ($(IZ) + (-0.625cm,-1.5cm)$);

    \node[circle,fill,inner sep=1.25pt, outer sep=1pt, label = {[label distance=0.15cm]-135:$\underline{Z}$}] at (IZ) (nIZ) {};
    \node[circle,fill,inner sep=1.25pt, outer sep=1pt, label = north:$\underline{X}$] at (IX) (nIX) {};
    \node[circle,fill,inner sep=1.25pt, outer sep=1pt, label = north:$\underline{K}$] at (IK) (nIK) {}; 
    \node[draw=none] at (ICI) (nICI) {$\CI'$}; 

    \path
     (nIZ) edge [out=180, in=-90, thick, red] node[below, yshift=-0.1cm] {$[2]$} node[left, yshift=0cm,black] {$a_2$} (nIX)
     (nIX) edge [thick, red] node[above, yshift=0.1cm] {$[3]$} node[right, yshift=0cm,black] {$b_1$} (nIZ)
     (nIZ) edge [] node[right] {$a_4$} (nIK)
     (nIK) edge [] node[above] {$d_2$} (nIX)
     ;

    \coordinate (IIZ) at ($(IZ) + (3cm,0)$);
    \coordinate (IIX) at ($(IIZ) + (-1.25cm,1.25cm)$);
    \coordinate (IIK) at ($(IIZ) + (0,1.25cm)$);
    \coordinate (IICI) at ($(IIZ) + (-0.625cm,-1.5cm)$);

    \node[circle,fill,inner sep=1.25pt, outer sep=1pt, label = {[label distance=0.15cm]-135:$\underline{Z}$}] at (IIZ) (nIIZ) {};
    \node[circle,fill,inner sep=1.25pt, outer sep=1pt, label = north:$\underline{X}$] at (IIX) (nIIX) {};
    \node[circle,fill,inner sep=1.25pt, outer sep=1pt, label = north:$\underline{K}$] at (IIK) (nIIK) {}; 
    \node[draw=none] at (IICI) (nIICI) {$\CI''$}; 

    \path
     (nIIX) edge [thick, red] node[above, yshift=0.1cm] {$[3]$} node[right, yshift=0cm,black] {$b_1$} (nIIZ)
     (nIIZ) edge [] node[right] {$a_4$} (nIIK)
     (nIIK) edge [] node[above] {$d_2$} (nIIX)
     ;

     \coordinate (IIIZ) at ($(IIZ) + (3.5cm,0)$);
     \coordinate (IIIX) at ($(IIIZ) + (-1.25cm,1.25cm)$);
     \coordinate (IIICI) at ($(IIIZ) + (-1.25cm,-1.5cm)$);
    
    \node[circle,fill,inner sep=1.25pt, outer sep=1pt, label = north:$\underline{X}$] at (IIIX) (nIIIX) {};
    \node[draw=none] at (IIICI) (nIIICI) {$\CI^{(3)}$};
    
    \path
     ;

    \coordinate (IVZ) at ($(Z) + (-4cm,0)$);
    \coordinate (IVX) at ($(IVZ) + (-1.25cm,1.25cm)$);
    \coordinate (IVY) at ($(IVZ) + (1.25cm,1.25cm)$);
    \coordinate (IVK) at ($(IVZ) + (0,1.25cm)$);
    \coordinate (IVCI) at ($(IVZ) + (0,-1.5cm)$);

    \node[circle,fill,inner sep=1.25pt, outer sep=1pt, label = {[label distance=0.15cm]-135:$\underline{Z}$}] at (IVZ) (nIVZ) {};
    \node[circle,fill,inner sep=1.25pt, outer sep=1pt, label = north:$\underline{X}$] at (IVX) (nIVX) {};
    \node[circle,fill,inner sep=1.25pt, outer sep=1pt, label = north:$\underline{Y}$] at (IVY) (nIVY) {}; 
    \node[circle,fill,inner sep=1.25pt, outer sep=1pt, label = north:$\underline{K}$] at (IVK) (nIVK) {}; 
    \node[draw=none] at (IVCI) (nIVCI) {$\hat{\CI}'/\CI$}; 

    \path
     (nIVZ) edge [out=180, in=-90, thick, red] node[below, yshift=-0.1cm] {$[2]$} node[left, yshift=0cm,black] {$a_2$} (nIVX)
     (nIVX) edge [] node[above, yshift=0.1cm] {} node[right, yshift=0cm,black] {$b_1$} (nIVZ)
     (nIVZ) edge [out=0, in=-90, thick, red] node[yshift=-0.35cm] {$[1]$} node[right, yshift=0cm,black] {$a_3$} (nIVY)
     (nIVY) edge [] node[right, yshift=0cm,black] {$d_1$} (nIVZ)
     (nIVZ) edge [out=-45,in=-135,loop, thick, red] node[right,xshift=0.1cm] {$[1]$} node[below, black] {$a_1$} (nIVZ)
     (nIVZ) edge [thick,red] node[above right] {$[3]$} node[right,black] {$a_4$} (nIVK)
     (nIVK) edge [] node[above] {$d_2$} (nIVX)
     ;

    \coordinate (P1) at ($(Z) + (1cm, -0.75cm)$);
    \coordinate (P2) at ($(IZ) + (-1cm, -0.75cm)$);
    \coordinate (P3) at ($(IZ) + (0.5cm, -0.75cm)$);
    \coordinate (P4) at ($(IIZ) + (-1cm, -0.75cm)$);
    \coordinate (P5) at ($(IIZ) + (0.5cm, -0.75cm)$);
    \coordinate (P6) at ($(IIIZ) + (-1.75cm, -0.75cm)$);

    \draw [thick, -To, line width=1.5pt, out=-30,in=210] (P1) to node[below] {step 1} 
    (P2);
    \draw [thick, -To, line width=1.5pt, out=-30,in=210] (P3) to node[below] {step 2} (P4);
    \draw [thick, -To, line width=1.5pt, out=-30,in=210] (P5) to node[below] {step 3} (P6);
    
    \end{tikzpicture}
    \\
    \end{tabular}
    } 
        \caption{An example of a non-LLEE LEE witness}
        \label{figure:An example of a non-LLEE LEE witness}
\end{figure}
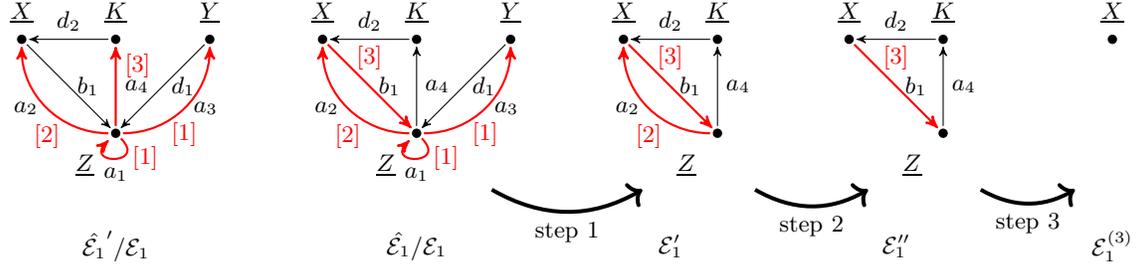

In an LLEE chart, a relation $X\lpb Y$
is defined if 
there exists a path from $X$ to $Y$ that begins with a loop-entry transition and continues through subsequent body transitions without reaching $X$ again. 
Or formally, 
$X\xrightarrow[\not\eqvn X]{}_{[m]} \cdot \xrightarrow[\not\eqvn X]{}^*_{\bo} Y$ for some $m > 0$. 
The transitive closure $X\lpb^+ Y$ is defined s.t. (i) $X\lpb Y$, or (ii) $X\lpb^+ Z$ and $Z\lpb Y$ for some node $Z$. 

The following property on relation $\lpb$ directly comes from the LLEE property. Refer to~\cite{Grabmayer20} for more details. 

\begin{proposition}
\label{prop:looping-back relation}
    Relation $\lpb^+$ in an LLEE chart $\mcl{G}$ is a well-founded, strict partially-ordered relation. 
\end{proposition}

\paragraph{Sovability of LLEE Charts (*)}
As stated in Sect.~\ref{section:Introduction}, 
LLEE charts characterize regular expressions modulo bisimulation equivalence $=$ and 
possess the capability of deriving unique 1-free-expression solutions under \MilIfree. 
We list the following results obtained from~\cite{Grabmayer20}. 
They correspond to the facts (F1) - (F4) respectively. 
As they are irrelevant to stage (II) and so are not concerned with our proof method proposed in this paper, we omit their proofs and interesting readers can refer to~\cite{Grabmayer20} for more details. 

\ifx
\begin{prop}[cf. Lemma 5.5 in~\cite{???}]
\label{prop:Derivability of LLEE-1-chart}
    Given an LLEE 1-chart $\mcl{H}$ and a loop sub-1-chart $\mcl{A}$ of $\mcl{H}$ starting from node $X$, if $X\pincycle Y$ for some $Y\in \mcl{A}$, then from $\vdash s(\mcl{A})$, we can derive $\vdash s(Y) = e\cdot s(X)$, where $e$ is a regular expression independent from $s$. 
\end{prop}
\fi


\begin{proposition}[Corresponding to (F4)]
Given two charts $\mcl{G}=\la V_1, A_1, v_1, \tran_1, \surd\ra$, 
$\mcl{H}=\la V_2, A_2, v_2, \tran_2, \surd\ra$ and a bisimulation function $\theta: V_1\to V_2$ between them, 
    if $s: V_2\to \IFSExp$ is a provable solution of $\mcl{H}$, then $s\circ \theta: V_1\to \IFSExp$ is a provable solution of $\mcl{G}$.  
\end{proposition}

\begin{proposition}[Corresponding to (F2)]
\label{prop:from 1-free regular expressions to LLEE charts}
	For every 1-free regular expression $e$, there is an LLEE chart $\hat{\mclul{G}}/\mclul{G}$ with an initial node $X$ such that $\mcl{G}$ has a provable solution $s$ with $s(X)\dddef e$.  
\end{proposition}

\begin{proposition}[Corresponding to (F1) and (F3)]
\label{prop:from LLEE charts to 1-free regular expressions}
	Every LLEE chart $\hat{\mclul{G}}/\mcl{G}$ has a provable solution $s$. 
	For any provable solutions $s_1$ and $s_2$ of $\mcl{G}$, $\IFvdash s_1(X) = s_2(X)$ for any node $X$ of $\mcl{G}$. 
\end{proposition}

\section{Looping-back Charts and Images}
\label{section:Looping-back Charts and well-structured Images}

As illustrated in Fig.~\ref{figure:Minimization strategy} of Sect.~\ref{section:Introduction}, our 
proof method focuses on stage (II') of the minimization strategy, which concerns 
how to acquire an LLEE bisimulation collapse $\mcl{H}$ from one of two 
bisimulation functions: $\theta_1$ from $\mcl{G}_{e_1}$ to $\mcl{H}$ and 
$\theta_2$ from $\mcl{G}_{e_2}$ to $\mcl{H}$. 
In this section, we introduce the basic concepts needed for the proofs  given in Sect.~\ref{section:Restoration for LLEE Charts}. 

\subsection{Some Conventions of Notations}
\label{section:Some Conventions of Notations}
In the rest of the paper, for convenience, we use the conventions of notations stated as follows.

Given a chart $\mcl{G}$, we often use $V_\mcl{G}$, $A_\mcl{G}$, $s_\mcl{G}$ and $\tran_{\mcl{G}}$ to represent its components. 
We usually ignore the initial state $s_\mcl{G}\in V_{\mcl{G}}$ and simply understand 
a chart $\la V_\mcl{G}, A_\mcl{G}, s_\mcl{G}, \tran_\mcl{G}, \surd\ra$ as a 4-tuple: $\la V_\mcl{G}, A_\mcl{G}, \tran_\mcl{G}, \surd\ra$.
This is because in our work, 
we do not need to distinguish the initial nodes to relate to the concepts like
chart interpretation and primary solution. 

Given a set of nodes $N\subseteq V_\mcl{G}$ of a chart $\mcl{G}$, 
$\gra_{\mcl{G}}(N)$ (or simply $\gra(N)$ when $\mcl{G}$ is clear in the context) denotes the \emph{chart of a set $N$ of nodes} 
which consists of all nodes and all transitions between these nodes. 
Formally, 
$\gra_{\mcl{G}}(N)\dddef \la N, A_\mcl{G}, \tran, \surd\ra$,  
where 
$\tran = \{X\xrightarrow{a} Y\ |\ X, Y\in N, a\in A_\mcl{G}\}$. 
\ifx
where 
$\tran = \{X\xrightarrow{a} Y\ |\ X, Y\in N, a\in A_\mcl{G}\}
\cup \{X\xrightarrow{a}\surd\ |\ X\in N, a\in A_\mcl{G}\}$. 
\fi
$\gra(\cdot)$ is often used as an abbreviation of a sub-chart by its set of nodes.

We define $\mcl{A}_1\unionc \mcl{A}_2 \dddef \gra_{\mcl{G}}(V_{\mcl{A}_1}\cup V_{\mcl{A}_2})$ as the chart of the union of the nodes of $\mcl{A}_1$ and $\mcl{A}_2$, provided that $\mcl{A}_1$ and $\mcl{A}_2$ are sub-charts of a chart $\mcl{G}$.  
We write $\mcl{G}\subc \mcl{H}$ (resp. $\mcl{G}\psubc \mcl{H}$) if 
$\mcl{G}$ is a sub-chart (resp. proper sub-chart) of $\mcl{H}$,  
and write $\mcl{G}\eqvc \mcl{H}$ if charts $\mcl{G}$ and $\mcl{H}$ are identical. 

Given a bisimulation function $\theta: V_{\mcl{G}}\to V_{\mcl{H}}$ from chart $\mcl{G}$ to $\mcl{H}$ and a sub-chart $\mcl{A}$ of $\mcl{G}$, 
we use $\theta(\mcl{A})$ to represent the sub-chart $\gra_{\mcl{H}}(\theta(V_{\mcl{A}}))$ of the set $\theta(V_{\mcl{A}})$ of nodes of $\mcl{H}$.  

We usually use small alphabets like $x, y, z, ...$ to represent nodes in a domain $V_\mcl{G}$, while using capital alphabets like $X, Y, Z, ...$ to express nodes in a co-domain $V_\mcl{H}$, of a bisimulation relation $\R:V_\mcl{G}\to V_\mcl{H}$ from $\mcl{G}$ to $\mcl{H}$. And we use the small and capital alphabets with the same name and subscript, but different superscripts to express bisimilar nodes in these two graphs. For example, we can write bisimilar-nodes as: 
$x, x', x'', x^{(3)}$ in $\mcl{G}$, and write bisimilar nodes $X, X', X'', X^{(3)}$ in $\mcl{H}$, which are bisimilar to $x$.

\subsection{Looping-back Charts}
We first introduce an LLEE sub-structure called \emph{looping-back charts}. 
It together with its images defined in Sect.~\ref{section:well-structured Images} plays the central role in our method.

\begin{definition}[Looping-back Chart]
\label{def:Looping-back Chart}
\ifx
A looping-back chart $\LC_x$ of an LLEE chart $\hat{\mcl{G}}/\mcl{G}$ w.r.t. a node $x$ is defined as the 
$\la x, \E_x\ra$-generated chart, where 
$\E_x = \{x\xrightarrow{}_{[m]}y\ |\ y\in V_{\mcl{G}}, m > 0\}$ is the set of all loop-entry transitions starting from $x$. 

 Also call $x$ the ``start node'' of $\LC_x$. 
\fi
The looping-back chart of an LLEE chart $\hat{\mcl{G}}/\mcl{G}$ w.r.t. a node $x$, denoted by $\LC_x$, is a chart defined as
$\LC_x\dddef \gra_{\mcl{G}}(\{x\}\cup \{y\ |\ x\lpb^+ y\})$ and 
satisfying that $\LC_x$ contains at least one loop. 

\ifx
where 
$V=(\{x\}\cup \{y\ |\ x\lpb^+ y\})\subseteq V_\mcl{G}$, 
and 
$\tran = \{X\xrightarrow{a} Y\ |\ X, Y\in V, a\in A_\mcl{G}\}\cup \{Y\xrightarrow{a} \surd\ |\ Y\in V/\{x\}\}$, and $\to\neq \emptyset$.   
\fi

Call $x$ the ``start node'' of $\LC_x$. 

\end{definition}

Note that the stipulation that $\LC_x$ must contain at least one loop forbids trivial charts of the form: $\la \{x\}, A, x, \emptyset, \surd\ra$ as looping-back charts, which contain only one node $x$ without any transitions. 


A sub-chart $\mcl{B}$ of a looping-back chart $\mcl{A}$ is also called a \emph{looping-back sub-chart} if it is a looping-back chart. 
Call $V_{\LC_x}\setminus\{x\}$ the \emph{body} of looping-back chart $\LC_x$. 


For instance,
in the LLEE chart $\hat{\CII}/\CII$ of Fig.~\ref{figure:An example of images}, charts $\gra(\{z',x'\})_{z'}$, $\gra(\{z'', y\})_{z''}$, $\gra(\{x, z'', k, y\})_x$
and the whole chart $(\CII)_z$ are looping-back charts. Among them note that charts $\gra(\{z',x'\})_{z'}$ and $\gra(\{z'', y\})_{z''}$ are also loop sub-charts of $\CII$.

\begin{figure}[tbhp]
    \centering
    \noindent\makebox[\textwidth]{ 
    \begin{tabular}{c}
    \begin{tikzpicture}[->,>=stealth', node distance=1.7cm]

     \coordinate (z) at (0,0);
     \coordinate (x) at ($(z) + (1cm,1.25cm)$);
     \coordinate (zi) at ($(z) + (-1cm,1.25cm)$);
     \coordinate (xi) at ($(zi) + (0, 1.25cm)$);
     \coordinate (k) at ($(z) + (0,2.5cm)$);
     \coordinate (zii) at ($(x) + (0,1.25cm)$);
     \coordinate (y) at ($(zii) + (0,1.25cm)$);
     \coordinate (CII) at ($(z) + (0,-1.5cm)$);
     
    \node[circle,fill,inner sep=1.25pt, outer sep=1pt, label = south:$\underline{z}$] at (z) (nz) {};
    \node[circle,fill,inner sep=1.25pt, outer sep=1pt, label = east:$\underline{x}$] at (x) (nx) {};
    \node[circle,fill,inner sep=1.25pt, outer sep=1pt, label = west:$\underline{z'}$] at (zi) (nzi) {};
    \node[circle,fill,inner sep=1.25pt, outer sep=1pt, label = west:$\underline{x'}$] at (xi) (nxi) {};
    \node[circle,fill,inner sep=1.25pt, outer sep=1pt, label = north:$\underline{k}$] at (k) (nk) {}; 
    \node[circle,fill,inner sep=1.25pt, outer sep=1pt, label = {[label distance=0.2cm]-5:$\underline{z''}$}] at (zii) (nzii) {};
    \node[circle,fill,inner sep=1.25pt, outer sep=1pt, label = north:$\underline{y}$] at (y) (ny) {};
    \node[draw=none] at (CII) (nCII) {$\hat{\CII}/\CII$};

     \path
     (nz) edge [out=0, in=-90, thick, red] node[below, yshift=-0.1cm] {$[3]$} node[right, yshift=0cm,black] {$a_2$} (nx)
     (nz) edge [out=180, in=-90, thick, red] node[left] {$[2]$} node[below right, yshift=-0.2cm,black] {$a_1$} (nzi)
     (nx) edge [] node[above right, xshift=-0.3cm, yshift=0.1cm, black] {$b_1$} (nz)
     (nzi) edge [] node[right, yshift=0.2cm, xshift=-0.25cm, black] {$a_1$} (nz)
     (nzi) edge [out=135, in=225, thick, red] node[right,xshift=-0.1cm] {$[1]$} node[left, yshift=0cm,black, xshift=0.1cm] {$a_2$} (nxi)
     (nxi) edge [out=-45, in=45] node[right, xshift=-0.1cm,black] {$b_1$} (nzi)
     (nx) edge [out=135, in=225, thick, red] node[right,xshift=-0.1cm] {$[2]$} node[left, yshift=0.4cm, xshift=0.2cm, black] {$b_1$} (nzii)
     (nzii) edge [out=-45, in=45] node[right, yshift=-0.15cm, xshift=-0.1cm, black] {$a_2$} (nx)
     (nzii) edge [out=135, in=225, thick, red] node[right,xshift=-0.1cm] {$[1]$} node[left, yshift=0cm,black,xshift=0.1cm] {$a_3$} (ny)
     (ny) edge [out=-45, in=45] node[right, yshift=0cm,black,xshift=-0.1cm, yshift=0.25cm] {$d_1$} (nzii)
     (nz) edge[out=-30,in=0,looseness=1.2, thick, red] node[right, yshift=1cm, xshift=0.1cm] {$[3]$} node[right,black,xshift=-0.1cm] {$a_3$} (ny)
     (nzii) edge[out=30,in=-30,loop, thick, red] node[above, yshift=0.15cm, xshift=0.25cm] {$[1]$} node[right,black, yshift=0.1cm] {$a_1$} (nzii)
     (nzii) edge[] node[above,yshift=-0.1cm] {$a_4$} (nk)
     (nk) edge[out=-90,in=180] node[right,xshift=-0.1cm,yshift=0.1cm] {$d_2$} (nx)
     (nz) edge[out=105,in=-110, thick, red] node[right,yshift=-0.5cm] {$[3]$} node[right,xshift=-0.1cm,black] {$a_4$} (nk)
    ;
    \draw (nzi) to [out=0,in=-90] ($(k) + (-0.5cm,0)$) to [out=90,in=180] node[above] {$a_3$} (ny);
    \draw (nzi) to [out=0,in=-90] ($(k) + (-0.35cm,-0.625cm)$) to [out=90,in=180] node[left, xshift=0.3cm, yshift=0.45cm] {$a_4$} (nk);

    \coordinate (Z) at ($(z) + (-5cm,0cm)$);
    \coordinate (X) at ($(Z) + (-1.25cm,1.25cm)$);
    \coordinate (Y) at ($(Z) + (1.25cm,1.25cm)$);
    \coordinate (K) at ($(Z) + (0,1.25cm)$);
    \coordinate (CI) at ($(Z) + (0,-1.5cm)$);

    \node[circle,fill,inner sep=1.25pt, outer sep=1pt, label = {[label distance=0.15cm]-135:$\underline{Z}$}] at (Z) (nZ) {};
    \node[circle,fill,inner sep=1.25pt, outer sep=1pt, label = north:$\underline{X}$] at (X) (nX) {};
    \node[circle,fill,inner sep=1.25pt, outer sep=1pt, label = north:$\underline{Y}$] at (Y) (nY) {}; 
    \node[circle,fill,inner sep=1.25pt, outer sep=1pt, label = north:$\underline{K}$] at (K) (nK) {}; 
    \node[draw=none] at (CI) (nCI) {$\CI$}; 

    \path
     (nZ) edge [out=180, in=-90] node[below, yshift=-0.1cm] {} node[left, yshift=0cm,black] {$a_2$} (nX)
     (nX) edge [] node[above, yshift=0.1cm] {} node[right, yshift=0cm,black] {$b_1$} (nZ)
     (nZ) edge [out=0, in=-90] node[yshift=-0.35cm] {} node[right, yshift=0cm,black] {$a_3$} (nY)
     (nY) edge [] node[right, yshift=0cm,black] {$d_1$} (nZ)
     (nZ) edge [out=-45,in=-135,loop] node[right,xshift=0.1cm] {} node[below] {$a_1$} (nZ)
     (nZ) edge [] node[right] {$a_4$} (nK)
     (nK) edge [] node[above] {$d_2$} (nX)
     ;
     
    \coordinate (P1) at ($(z) + (-1cm,-0.5cm)$);
    \coordinate (P2) at ($(Z) + (1cm,-0.5cm)$);
    
    \draw[-To,dashed, very thick,out=210,in=-30] (P1) to node[above] {$\theta$} (P2); 

    \ifx
    \coordinate (Imgi) at ($(Z) + (-4cm,0)$);
    \coordinate (Imgii) at ($(Imgi) + (-1.25cm, 1.25cm)$);
    \coordinate (Imgiii) at ($(Imgi) + (1.25cm, 1.25cm)$);
    \coordinate (eoi) at ($(Imgi) + (-0.625cm,0.625cm)$);

    \node[draw=none] at (Imgi) (nImgi) {$\CI$};
    \node[draw=none] at (Imgii) (nImgii) {$\{Z, X\}$};
    \node[draw=none] at (Imgiii) (nImgiii) {$\{Z, Y\}$};
    \node[draw=none, rotate=-45] at (eoi) (neoi) {$\eo$};
    \fi

    \ifx
    \coordinate (order) at ($(Z) + (-4cm,0.625cm)$);

    \node[draw=none] at (order) (norder) {$\CI\eo \{Z, X\}\eo \{Z, Y\}$};
    \fi

    \coordinate (Zi) at ($(Z) + (-5cm,0cm)$);
    \coordinate (Xi) at ($(Zi) + (-1.25cm,1.25cm)$);
    \coordinate (Yi) at ($(Zi) + (1.25cm,1.25cm)$);
    \coordinate (Ki) at ($(Zi) + (0,1.25cm)$);
    \coordinate (CIi) at ($(Zi) + (0,-1.5cm)$);

    \node[circle,fill,inner sep=1.25pt, outer sep=1pt, label = {[label distance=0.15cm]-135:$\underline{Z}$}] at (Zi) (nZi) {};
    \node[circle,fill,inner sep=1.25pt, outer sep=1pt, label = north:$\underline{X}$] at (Xi) (nXi) {};
    \node[circle,fill,inner sep=1.25pt, outer sep=1pt, label = north:$\underline{Y}$] at (Yi) (nYi) {}; 
    \node[circle,fill,inner sep=1.25pt, outer sep=1pt, label = north:$\underline{K}$] at (Ki) (nKi) {}; 
    \node[draw=none] at (CIi) (nCIi) {$\hat{\CI}/\CI$}; 

    \path
     (nZi) edge [out=180, in=-90, thick, red] node[below, yshift=-0.1cm] {$[2]$} node[left, yshift=0cm,black] {$a_2$} (nXi)
     (nXi) edge [thick, red] node[above, yshift=0.1cm] {$[3]$} node[right, yshift=0cm,black] {$b_1$} (nZi)
     (nZi) edge [out=0, in=-90, thick, red] node[yshift=-0.35cm] {$[1]$} node[right, yshift=0cm,black] {$a_3$} (nYi)
     (nYi) edge [] node[right, yshift=0cm,black] {$d_1$} (nZi)
     (nZi) edge [out=-45,in=-135,loop, thick, red] node[right,xshift=0.1cm] {$[1]$} node[below, black] {$a_1$} (nZi)
     (nZi) edge [] node[right] {$a_4$} (nKi)
     (nKi) edge [] node[above] {$d_2$} (nXi)
     ;

    \end{tikzpicture}
    \\
    \end{tabular}
    } 
        \caption{An example of images}
        \label{figure:An example of images}
\end{figure}
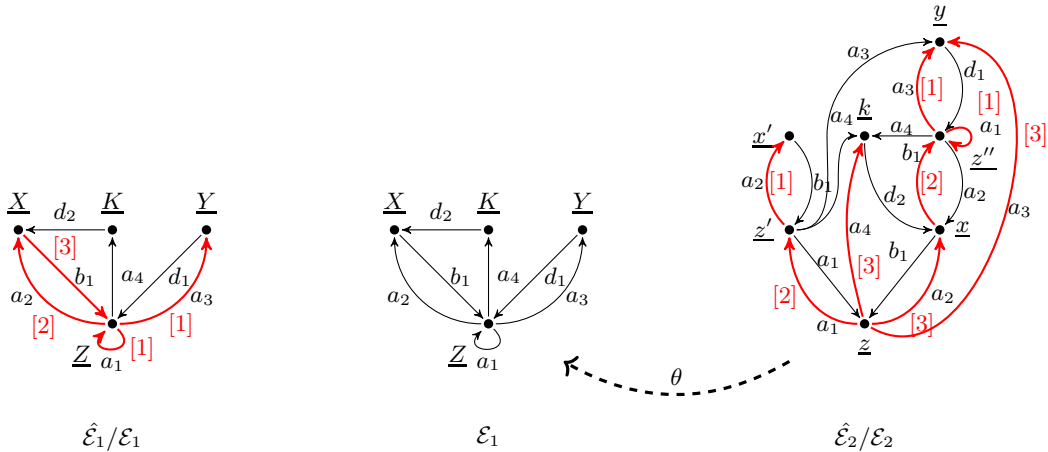

The following properties about looping-back charts are straightforward 
from LLEE charts. 

\begin{proposition}
\label{prop:looping-back chart}
    Given a looping-back chart $\mcl{A}_X$ of an LLEE chart $\hat{\mcl{G}}/\mcl{G}$, the following propositions hold:
    
\begin{enumerate}[(i)]
    
    
    \item 
    \label{item: looping-back property - looping-back sub-chart}
    For any node $Y\in V_{\mcl{A}_X}\setminus\{X\}$ and the looping-back chart $\mcl{B}_Y$ (if there is one), 
$\mcl{B}_Y\psubc \mcl{A}_X$;
    

    \item 
    \label{item: looping-back property - path}
    There is no transition $Y\xrightarrow{} Z$ such that $Y\in V_{\mcl{A}_X}\setminus \{X\}$ and $Z\notin V_{\mcl{A}_X}$. 

    \item 
    \label{item: looping-back property - termination}
    Any node $Y\in V_{\mcl{A}_X}\setminus\{X\}$ does not reach $\surd$ before returning to $X$. 
    
\end{enumerate}
    
\end{proposition}

 
\begin{proof}

\ifx
(\ref{item: looping-back property - loop sub-chart}):
By the definition of loop sub-chart in Sect.~\ref{section:LEE-LLEE Charts}, 
each node $Y$ in $\mcl{B}^{\ip}$ satisfies $X\lpb Y$. 
So $Y\in V_{\mcl{A}_X}$ by Def.~\ref{def:Looping-back Chart}. 
Since $V_\mcl{B} \subseteq V_{\mcl{A}_X}$, by Def.~\ref{def:Looping-back Chart} it is not hard to see $\mcl{B}\subc \mcl{A}_X$. 
\fi

(\ref{item: looping-back property - looping-back sub-chart}): 
By the transitivity of $\lpb^+$ indicated by Prop.~\ref{prop:looping-back relation}, 
for any $Z$ such that $Y\lpb^+ Z$, since $X\lpb^+ Y$, $X\lpb^+ Z$. 
So $V_{\mcl{B}_Y}\subseteq V_{\mcl{A}_X}$. 
By Def.~\ref{def:Looping-back Chart}, $\mcl{B}_Y\subc \mcl{A}_X$. 
    If $\mcl{B}_Y \eqvc \mcl{A}_X$, since $Y\not\!\!\eqvn X$,  
Prop.~\ref{prop:looping-back relation} is violated. 
Because we have both $X\lpb^+ Y$ and $Y\lpb^+ X$. 
Hence $\mcl{B}_Y\psubc \mcl{A}_X$. 

(\ref{item: looping-back property - path}): 
Assume a transition $Y\xrightarrow{} Z$ with $Y\in V_{\mcl{A}_X}\setminus \{X\}$ but $Z\notin V_{\mcl{A}_X}$. 
Let $X\lpb^+ K\lpb Y$ for some node $K$. 
Then by the definition of relation $\lpb$ (Sect.~\ref{section:LEE-LLEE Charts}), 
depending on whether $Y\xrightarrow{} Z$ is a loop-entry or a body transition, 
we have either $K\lpb Y\lpb Z$ or $K\lpb Z$. 
So $X\lpb^+ K\lpb^+ Z$. 
But by Def.~\ref{prop:looping-back chart}, $Z\in V_{\mcl{A}_X}$. This leads to contradiction. 

(\ref{item: looping-back property - termination}):
By Def.~\ref{def:Looping-back Chart}, for any node $Y\in V_{\mcl{A}_X}\setminus \{X\}$, 
$X\lpb^+ Y$. 
Proceed by induction on relation $\lpb^+$. 

\textit{Base case:}
Consider a node $Y$ with $X\lpb Y$. 
By the definition of relation $\lpb$ and the LEE property (Sect.~\ref{section:LEE-LLEE Charts}), $Y$ is in the body of a sub-chart that is or will become a loop sub-chart starting from $X$ during the elimination process indicated by $\hat{\mcl{G}}$. 
By (L3), $Y$ cannot reach $\surd$ before returning to $X$. 

\textit{Step case:}
Consider a node $Y$ such that $X\lpb^+ K\lpb Y$ for some node $K$. 
Similarly, by the definition of relation $\lpb$ and the LEE property, 
$Y$ is in the body of a sub-chart that is or will become a loop sub-chart starting from $K$ during the elimination process indicated by $\hat{\mcl{G}}$.  
By (L3), $Y$ cannot reach $\surd$ before returning to $K$. 
By induction hypothesis, $K$ cannot reach $\surd$ before returning to $X$. 
Therefore, $Y$ cannot reach $\surd$ before returning to $X$. 
\qed

\ifx 
By Def.~\ref{def:Looping-back Chart} and (\ref{item: looping-back property - path}) above, 
it is sufficient to prove that for any $Y$ such that $X\lpb^+ Y$, 
$Y$ cannot reach $\surd$ before returning to $X$. 
Proceed by induction on the structure pf $\mcl{A}_X$ indicated by sub-chart relations. 
In the base case, let $\mcl{A}_X$ be a looping-back chart without any looping-back proper sub-charts. For any $Y\in V_{\mcl{A}_X}/\{X\}$, it must be $X\lpb Y$. 
By the definition of $\lpb$ (Sect.~\ref{section:LEE-LLEE Charts}), it is easy to see that $\mcl{A}_X$ is in fact a loop sub-chart of $\mcl{G}$. 
So by (L3), $Y$ cannot reach $\surd$ before returning to $X$. 
In the step case, consider an arbitrary $\mcl{A}_X$. 
By inductive hypothesis, for any node $Y$ in a looping-back proper sub-chart $\mcl{B}_Z$ of $\mcl{A}_X$ such that $Z\lpb^+ Y$, $Y$ cannot reach $\surd$ before returning to $Z$. 

In the base case, consider a node $Y$ with $X\lpb Y$. 
Since $Y$ is in the body of a loop sub-chart starting from $X$, 
by (L3), $Y$ cannot reach $\surd$ before returning to $X$.
In the step case, consider a node $Y$ such that $X\lpb^+ K\lpb Y$. 
Then $Y$ is in the body of a loop sub-chart starting from $K$. 
By (L3) $Y$ cannot reach $\surd$ before returning to $K$. 
By induction hypothesis, $K$ cannot reach $\surd$ before returning to $X$. 
Therefore, we obtain the result.  
\fi

\ifx
By Def.~\ref{def:Looping-back Chart} and (\ref{item: looping-back property - path}) above, 
it is sufficient to prove that for any $Y$ such that $X\lpb^+ Y$, 
$Y$ cannot reach $\surd$ before returning to $X$. 
Proceed by induction on relation $\lpb^+$. 
In the base case, consider a node $Y$ with $X\lpb Y$. 
Since $Y$ is in the body of a loop sub-chart starting from $X$, 
by (L3), $Y$ cannot reach $\surd$ before returning to $X$.
In the step case, consider a node $Y$ such that $X\lpb^+ K\lpb Y$. 
Then $Y$ is in the body of a loop sub-chart starting from $K$. 
By (L3) $Y$ cannot reach $\surd$ before returning to $K$. 
By induction hypothesis, $K$ cannot reach $\surd$ before returning to $X$. 
Therefore, we obtain the result.  
\fi

\ifx
(\ref{item: looping-back property - loop}): 
   If a loop $\mcl{C}$ of $\mcl{A}_X$ does not pass 
$X$,  
let $Y\xrightarrow{}_{[m]}\cdot $ be the loop-entry transition of $\mcl{C}$ with $m > 0$ and $Y\not\!\!\eqvn X$. 
By Def.~\ref{def:Looping-back Chart}, $\mcl{C}$ is in the looping-back chart $\mcl{B}_Y$. 
Since $Y\in V_{\mcl{A}}/\{X\}$, by (\ref{item: looping-back property - looping-back sub-chart}) $\mcl{B}_Y\psubc \mcl{A}_X$. 
Hence $\mcl{C}\subc\mcl{B}_Y\psubc \mcl{A}_X$. 
\fi

\ifx
according to the LLEE property, 
$\mcl{C}$ is in some loop sub-chart $\mcl{B}$ of $\mcl{G}$, with some start node $Y$ in $\mcl{C}$. 
Since $\mcl{C}$ does not pass $X$, $Y\not\!\!\eqvn X$ and so $Y\in \mcl{A}^\ip$. 
By (\ref{item: looping-back property - loop sub-chart}) and (\ref{item: looping-back property - looping-back sub-chart}) above, 
$\mcl{B}$ is a sub-chart of the looping-back chart $\mcl{C}_Y$, which is a proper sub-chart of $\mcl{A}_X$. i.e., $\mcl{B}\subc \mcl{C}_Y\psubc \mcl{A}_X$. This immediately violates the assumption about $\mcl{C}$. 
\fi

\ifx
Because $\mcl{B}$ is not a sub-chart of $\mcl{A}$ (otherwise by (\ref{item: looping-back property - loop sub-chart}) above it would be a sub-chart of looping-back sub-chart of $\mcl{A}$, contradicting the assumption), to leave $\mcl{A}$, 
by (L2), there is a path $Y\xrightarrow{}_{[\cdot]}Z \xrightarrow{}^* X\xrightarrow{}...$ in $\mcl{B}$ starting from $Y$ and passing $X$.
But this means that $Y\lpb^+ X$, contradicting Prop.~\ref{prop:looping-back relation}. 
\fi

\ifx
But this immediately causes that $\mcl{G}$ violates the LLEE property. 
Because on one hand, any loop-entry transition in $\E_X$ is from the body of $\mcl{B}$; while on the other hand, the loop-entry transition $Y\xrightarrow{}_{[\cdot]}Z$ in $\E_Y$ is from the body of $\mcl{A}$. 
\fi
\end{proof}

\ifx
\begin{prop}
\label{prop:loop sub-chart 2}
    Given a looping-back chart $\mcl{A}(X, \E_X)$ of an LLEE chart $\hat{\mcl{G}}/\mcl{G}$, 
    each loop of $\mcl{A}(X, \E_X)$ either passes through $X$ or 
    is in some looping-back sub-chart of $\mcl{A}(X, \E_X)$. 
\end{prop}
\fi

\ifx
\begin{proof}
   Assume there exists a loop $\mcl{C}$ of $\mcl{A}(X, \E_X)$ that neither passes 
$X$ nor is in any looping-back sub-chart of $\mcl{A}$. 
Since $\mcl{G}$ is an LLEE chart, 
$\mcl{C}$ is in some loop sub-chart $\mcl{B}(Y, \E_Y)$ of $\mcl{G}$, with some start node $Y$ of $\mcl{C}$. 
Since $\mcl{C}$ does not pass $X$, $Y\not\equiv X$ and so $Y\in \mcl{A}^\ip$. 
Because $\mcl{B}$ is not a sub-chart of $\mcl{A}$, 
by (L2), there is a path $Y\xrightarrow{}_{[\cdot]}Z \xrightarrow{}...$ in $\mcl{B}$ starting from $Y$ which must pass $X$ to leave $\mcl{A}$. 
But this immediately causes that $\mcl{G}$ violates the LLEE property. 
Because on one hand, any loop-entry transition in $\E_X$ is from the body of $\mcl{B}$; while on the other hand, the loop-entry transition $Y\xrightarrow{}_{[\cdot]}Z$ in $\E_Y$ is from the body of $\mcl{A}$. 

\end{proof}
\fi

\subsection{Images and Well-structured Pre-images}
\label{section:well-structured Images}

\begin{definition}[Image]
	\label{def:Image on LLEE Charts}
	Given a bisimulation function $\theta: V_{\hat{\mcl{G}}/\mcl{G}} \to V_\mcl{H}$ from an LLEE chart $\hat{\mcl{G}}/\mcl{G}$ to a bismulation collapse $\mcl{H}$, 
 a sub-chart $\mcl{I}$ of $\mcl{H}$ is called an `image', 
 if there exists a looping-back chart $\LC_x$ in $\mcl{G}$ such that $\mcl{I}\eqvc\theta(\LC_x) = \gra_{\mcl{H}}(\theta(V_{\LC_x}))$. 

Also call $\mcl{I}$ ``the image of $\LC_x$'' for each 
such $\LC_x$, and call $\LC_x$ a `pre-image' of $\mcl{I}$. 


 Denote by $\mcl{I}_\theta(\mcl{H})$ the set of all images on $\mcl{H}$ (w.r.t. $\theta$). 
 
\end{definition}

Call $V_{\theta(\LC_x)}\setminus \{\theta(x)\}$ the 
\emph{body} of image $\theta(\LC_x)$ w.r.t. a pre-image $\LC_x$. 

An image can have more than one pre-image. 
Different images can overlap each other in the sense that they 
share same nodes and transitions. 

For example, Fig.~\ref{figure:An example of images} shows the bisimulation function $\theta: V_{\hat{\CII}/{\CII}}\to V_{\CI}$ from an LLEE chart $\hat{\CII}/\CII$ to its bisimulation collapse $\CI$. 
According to Sect.~\ref{section:Some Conventions of Notations}, 
the nodes with the same alphabets (e.g. $x, x'$ and $X$) are bisimilar to each other. 
In this example, 
there are totally 3 images on $\CI$:
$\CI$, $\gra(\{Z, X\})$ and $\gra(\{Z, Y\})$.
Among them image $\CI$ has both looping-back charts $({\CII})_z$ and $\gra(\{x, z'', k, y\})_x$ as its pre-images; 
image $\gra(\{Z, X\})$ has the pre-image $\gra(\{z', x'\})_{z'}$; 
image $\gra(\{Z, Y\})$ has the pre-image $\gra(\{z'', y\})_{z''}$. 
Images $\gra(\{Z, X\})$ and $\gra(\{Z, Y\})$ are proper sub-images of image $\CI$. 
$\gra(\{Z, X\})$ and $\gra(\{Z, Y\})$ are not sub-images of each other but overlap each other. 

\ifx
$\theta(\CII)$, $\theta(\{z, x'\}_{z'})$, $\theta(\{x, z'', k, y\}_{x})$ 
and $\theta(\{z'', y\}_{z''})$. 
Image $\theta(\{x, z'', k, y\}_{x})$ of the looping-back chart $\{x, z'', k, y\}_{x}$ and image $\theta(\{z', x'\}_{z'})$ of the looping-back chart $\{z', x'\}_{z'}$ overlap each other (with $\theta(\{x, z'', k, y\}) = \CI$ and $\theta(\{z', x'\}) = \{Z, X\}$). 
While image $\theta(\CII)$ of the looping-back chart $\CII$ is the same as the image $\theta(\{x, z'', k, y\}_{x})$ (with $\theta(\CII) = \theta(\{x, z'', k, y\}_{x}) = \CI$). 
\fi

A bisimulation function is monotonic w.r.t. sub-chart relation in the following sense. 

\begin{proposition}
\label{prop:monotonic of bisimlation function}
    A bisimulation function $\theta: V_{\hat{\mcl{G}}/\mcl{G}}\to V_\mcl{H}$ 
    satisfies that $\mcl{A}_1\subc \mcl{A}_2$ implies $\theta(\mcl{A}_1)\subc \theta(\mcl{A}_2)$ for any sub-charts $\mcl{A}_1$ and $\mcl{A}_2$ of $\mcl{G}$. 
\end{proposition}

\begin{proof}[Proof Sketch]
    $\mcl{A}_1\subc \mcl{A}_2$ means $V_{\mcl{A}_1}\subseteq V_{\mcl{A}_2}$. 
    So $\theta(V_{\mcl{A}_1})\subseteq \theta(V_{\mcl{A}_2})$. 
    By the definition of $\gra(\cdot)$ in Sect.~\ref{section:Some Conventions of Notations}, clearly $\gra(\theta(V_{\mcl{A}_1}))\subc \gra(\theta(V_{\mcl{A}_2}))$. 
    \qed
\end{proof}


\ifx
In this paper, we propose a way to manipulate these images so that we can build a strictly-partial order between them, according to which we thus manage to directly determine the LEE structure of the bisimulation collapse through a structure, we call \emph{region}.  
Below (of this sub-section and in Sect.~\ref{section:Elimination Order and Region}) we clarify our idea by introducing several critical concepts. 
\fi

Below we propose a critical property for images, called \emph{well-structuredness}, based on which we will show in Sect.~\ref{section:Restoration for LLEE Charts} that we manage to directly determine an LEE structure on images. 
\ifx
\begin{mydef}[Well-structured Image]
\label{def:well-structured Image}
    An image $\theta(\LC_x)$ w.r.t. a $\theta: \hat{\mcl{G}}/\mcl{G}\to \mcl{H}$ and a bisimulation collapse $\mcl{H}$ is called ``well constructed'', if 
    for any $\LC_y\subset \LC_x$ in $\mcl{G}$,
    $\theta(\LC_y)\subset \theta(\LC_x)$. 

    Denote by $\mcl{I}_\theta(\mcl{H})$ the set of all well-structured images on $\mcl{H}$ w.r.t. $\theta$. 
\end{mydef}
\fi

\ifx
The image $\CI$ of the looping-back chart $\CII$ of Fig.~\ref{figure:An example of images} is not well constructed, since $\{x, z'', k, y\}\subset \CII$ but 
$\theta(\{x, z'', k, y\}) = \CI = \theta(\CII)$. 
On the other hand, the (same) image $\CI$ of the looping-back chart $\{x, z'', k, y\}_x$ is well constructed. 
\fi

\begin{proposition}[Well-structuredness]
\label{prop:well-structured Image}
    Given a bisimulation function $\theta: V_{\hat{\mcl{G}}/\mcl{G}}\to V_\mcl{H}$ from an LLEE chart $\hat{\mcl{G}}/\mcl{G}$ to its bisimulation collapse $\mcl{H}$, 
    for any image $\mcl{I}$ of $\mcl{I}_\theta(\mcl{H})$,
    there exists a looping-back chart $\LC_x$ as a pre-image of  
    $\mcl{I}$ such that
    for any $\LC_y\psubc \LC_x$, 
    $\theta(\LC_y)\psubc \theta(\LC_x) \eqvc \mcl{I}$. 

    We say $\LC_x$ is `well structured' (w.r.t. $\theta$), and call $\LC_x$ a ``well-structured pre-image'' of image $\mcl{I}$. 
\end{proposition}

\ifx
\begin{proof}[Proof of Prop.~\ref{prop:well-structured Image} (Sketch)]
If a pre-image $\LC_x$ of $\mcl{I}$ is not well structured, 
which means there exists a proper looping-back sub-chart $\LC_y$ of $\LC_x$ such that $\theta(\LC_y)\eqvc \theta(\LC_x) \eqvc \mcl{I}$, 
then $\LC_y$ itself is a pre-image of $\mcl{I}$. 
By that $\subc$ is a well-founded relation, we can always obtain a well-structured pre-image for $\mcl{I}$.  
\qed
\end{proof}
\fi

\begin{proof}
If a pre-image $\LC_x$ of $\mcl{I}$ is not well structured, 
which means there exists a proper looping-back sub-chart $\LC_y$ of $\LC_x$ such that $\theta(\LC_y)\eqvc \theta(\LC_x) \eqvc \mcl{I}$, 
then $\LC_y$ itself is a pre-image of $\mcl{I}$. 
By that $\subc$ is a well-founded relation, we can always obtain a well-structured pre-image for $\mcl{I}$.  
\qed
\end{proof}

In Fig.~\ref{figure:An example of images}, the looping-back chart $({\CII})_z$, as a pre-image of image $\CI$, is not well structured, since 
$\gra(\{x, z'', k, y\})_x\psubc (\CII)_z$ but 
$\theta(\gra(\{x, z'', k, y\})_x) \eqvc \CI \eqvc \theta((\CII)_z)$. 
However, chart $\gra(\{x, z'', k, y\})_x$ is a well-structured pre-image of image $\CI$ w.r.t. $\theta$. 

Prop.~\ref{prop:well-structured Image} guides us to consider 
the relation between an image and its well-structured pre-images as looping-back charts.


\section{Image Reflection on LLEE Charts}
\label{section:Restoration for LLEE Charts}

\ifx
\subsection{Preparations of Main Proofs}

Before introducing and proving the main theorems Lemma~\ref{lemma:Image Reflection on LLEE Charts --- critical condition} and 
Theorem~\ref{theo:restoration for LLEE charts} of this paper, 
we first consider some facts about LEE/LLEE charts and bisimulation relations between charts. 
They are not (or not fully) addressed in previous work but are crucial for our work. 

An LLEE chart satisfies properties Prop.~\ref{prop:loop sub-chart} and Prop.~\ref{prop:loop sub-chart 2} as follows. 
They are direct results by the LLEE property and the definition of loop charts in Sect.~\ref{section:LEE-LLEE Charts}. 
(Note that although they can also be obtained as corollaries in early work~\cite{???}, here we give our own proofs without using auxiliary propositions. )

\begin{prop}
\label{prop:loop sub-chart}
    Given a looping-back chart $\mcl{A}(X, \E_X)$ of an LLEE chart $\hat{\mcl{G}}/\mcl{G}$, every other looping-back chart $\mcl{B}(Y, \E_Y)$ of $\mcl{A}(X, \E_X)$ satisfies that
    $\mcl{B}(Y, \E_Y)\subset \mcl{A}(X, \E_X)$. 
\end{prop}

\begin{proof}
    If $\mcl{B}(Y, \E_Y) = \mcl{A}(X, \E_E)$, then  
$\mcl{B}$ must be $\mcl{A}$ itself (i.e., $Y\equiv X$ and $\E_Y = \E_X$). 
Otherwise, LLEE property is violated. 
Because any loop-entry transitions $X\xrightarrow{}_{[\cdot ]}\cdot$ in $\E_X$ is from the body of $\mcl{B}$, and vice versa. 
\end{proof}

\begin{prop}
\label{prop:loop sub-chart 2}
    Given a loop sub-chart $\mcl{A}(X, \E_X)$ of an LLEE chart $\hat{\mcl{G}}/\mcl{G}$, 
    each loop of $\mcl{A}(X, \E_X)$ either passes through $X$ or 
    is in some loop sub-chart of $\mcl{A}(X, \E_X)$. 
\end{prop}


\begin{proof}
   Assume there exists a loop $\mcl{C}$ of $\mcl{A}(X, \E_X)$ that neither passes 
$X$ nor is in any loop sub-chart of $\mcl{A}$. 
Since $\mcl{G}$ is an LLEE chart, 
$\mcl{C}$ is in some loop sub-chart $\mcl{B}(Y, \E_Y)$ of $\mcl{G}$, with some start node $Y$ of $\mcl{C}$. 
Since $\mcl{C}$ does not pass $X$, $Y\not\equiv X$ and so $Y\in \mcl{A}^\ip$. 
Because $\mcl{B}$ is not a sub-chart of $\mcl{A}$, 
by (L2), there is a path $Y\xrightarrow{}_{[\cdot]}Z \xrightarrow{}...$ in $\mcl{B}$ starting from $Y$ which passes $X$ to leave $\mcl{A}$. 
But this immediately causes that $\mcl{G}$ violates the LLEE property. 
Because on one hand, any loop-entry transition in $\E_X$ is from the body of $\mcl{B}$; while on the other hand, the loop-entry transition $Y\xrightarrow{}_{[\cdot]}Z$ in $\E_Y$ is from the body of $\mcl{A}$. 
\end{proof}

\fi

\ifx
\begin{prop}
\label{prop:loop sub-chart}
    Given a loop sub-chart $\mcl{A}(X, \E_X)$ of an LLEE chart $\hat{\mcl{G}}/\mcl{G}$, for every loop sub-chart $\mcl{B}(Y, \E_Y)$ of $\mcl{G}$ where $Y\in \mcl{A}(X, \E_X)$, 
    $\mcl{B}(Y, \E_Y)\subset \mcl{A}(X, \E_X)$. 
\end{prop}
\fi

\ifx
Prop.~\ref{prop:loop sub-chart} says that in an LLEE chart $\mcl{G}$, 
all loop-entry transitions that need to be removed before the elimination of a loop sub-chart $\mcl{A}$ are exactly the loop-entry transitions of all loop sub-charts of $\mcl{A}$. 
The proof of Prop.~\ref{prop:loop sub-chart} can be obtained by observing that 
$\mcl{B}(Y, \E_Y)\cap \mcl{A}(X, E_X) \neq \{X\}$, because otherwise it violates 
the LLEE property when eliminating $\mcl{A}$. 
\fi


In this section, we give our proof of the completeness of \MilIfree. 
This is realized mainly by discussing the relation between images and their well-structured pre-images as looping-back charts. 
We call that 
an image can \emph{reflect} the LLEE structure of its well-structured pre-images, 
in the sense that the image has an LEE structure (called \emph{LEE sub-chart} as defined below) enforced by 
the LLEE sturcture of its well-structured pre-images.

The main result of this section is Theorem~\ref{theo:restoration for LLEE charts}, which relies on Prop.~\ref{prop:loop correspondence} and Lemma~\ref{lemma:Image Reflection on LLEE Charts --- critical condition}. 


\begin{definition}[LEE Sub-Chart]
\label{def:LEE Sub-Chart}
    A sub-chart $\mcl{A}$ of a chart $\mcl{G}$ is 
called an ``LEE sub-chart'' if it is an LEE chart. 
\end{definition}

In a chart, by the definition of loop sub-chart (Sect.~\ref{section:LEE-LLEE Charts}), every loop sub-chart is an LEE sub-chart. 
Note that a chart has an LEE sub-chart does not mean itself is an LEE chart. 

Given the bisimulation collapse $\mcl{H}$ of an LLEE chart $\mcl{G}$ with a bisimulation function $\theta: V_{\hat{\mcl{G}}/\mcl{G}}\to V_\mcl{H}$,  
the main idea of our proof is that through $\theta$, the behaviour of each image $\theta(\LC_x)$ of $\mcl{I}_\theta(\mcl{H})$ is constrained by the behaviour of its well-structured pre-image $\LC_x$ 
so that $\theta(\LC_x)$ is in fact an LEE sub-chart of $\mcl{H}$: 
By the sub-image relation between images, we show that there exists an elimination process by which each image $\theta(\LC_x)$ will eventually become a loop sub-chart itself and thus can be 
eliminated further. 
By eliminating all the images as LEE sub-charts of $\mcl{H}$, we show that there would be no loops left in $\mcl{H}$ because the images of $\mcl{I}_\theta(\mcl{H})$ cover all the loops of $\mcl{H}$. 

We use an example to illustrate our idea.  
Consider the chart $\CI$ (Fig.~\ref{figure:An example of a non-LLEE LEE witness}, Fig.~\ref{figure:An example of images}). 
The set of its images is $\mcl{I}_\theta(\CI) = \{\CI, \gra(\{Z, X\}), \gra(\{Z, Y\})\}$. 
To prove $\CI$ is an LEE chart, we conduct the following two steps of reasoning in an inside-out manner w.r.t. the sub-image relation:
\begin{itemize}
    \item Firstly, for the minimum image $\gra(\{Z, Y\}) \eqvc \theta(\gra(\{z'',y\})_{z''})$, 
we see that $\gra(\{Z, Y\})$ is a loop sub-chart (so is an LEE sub-chart) and can be eliminated in $\CI$:
As constrained by its well-structured pre-image $\gra(\{z'', y\})_{z''}$, both loops 
$Z\xrightarrow{a_1} Z$ and $Z\xrightarrow{a_3} Y\xrightarrow{d_1} Z$ pass through $Z$ and each path in $\gra(\{Z, Y\})$ from $Z$ stays in $\gra(\{Z, Y\})$ before returning to $Z$. 
Similarly, another minimum image $\gra(\{Z, X\})$ can also be eliminated because it is a loop sub-chart of $\CI$. 

\item Secondly, consider the image $\CI$ which has images $\gra(\{Z, Y\})$ and $\gra(\{Z, X\})$ as its two proper sub-images.  
After $\gra(\{Z, Y\})$ and $\gra(\{Z, X\})$ are eliminated (in an order for example first $\gra(\{Z, Y\})$ and then $\gra(\{Z, X\})$), 
the remnant chart: chart $\CI''$ of Fig.~\ref{figure:An example of a non-LLEE LEE witness}, becomes a loop sub-chart and thus can be eliminated. 
This is because its only loop $X\xrightarrow{b_1} Z\xrightarrow{a_4} K\xrightarrow{d_2} X$, constrained by the well-structured pre-image $\gra(\{x, z'', k, y\})_x$, passes through
$X$ and each path from $X$ stays in $\CI''$ before returning to $X$. 
Therefore, the image $\CI$ is an LEE sub-chart (of itself). 
\end{itemize}

Since $\mcl{I}_\theta(\CI)$ covers all the loops of $\CI$, 
after eliminating $\CI''$, we obtain the chart $\CI^{(3)}$ of Fig.~\ref{figure:An example of a non-LLEE LEE witness} where there is no loops. 
Therefore, the whole chart $\CI$ is an LEE chart.

Following the main idea, below we give our formal proofs. 

Prop.~\ref{prop:loop correspondence} indicates a correspondence between the loops 
of an LLEE chart $\mcl{G}$ and the loops of its bisimulation collapse $\mcl{H}$ given a bisimulation function $\theta: V_{\hat{\mcl{G}}/\mcl{G}}\to V_\mcl{H}$. 

\begin{proposition}
	\label{prop:loop correspondence}
	Given a bisimulation function $\theta: V_{\hat{\mcl{G}}/\mcl{G}}\to V_\mcl{H}$ from an LLEE chart $\mcl{G}$ to its bisimulation collapse $\mcl{H}$ and 
	a loop $\mcl{A}$ of $\mcl{H}$, 
	from a node $x$ of $\mcl{G}$ such that $\theta(x)$ is in $\mcl{A}$, there exists a path $\mcl{S}$ reaching a loop $\mcl{D}$ in $\mcl{G}$ satisfying that $\mcl{A} \eqvc \theta(\mcl{S}\unionc \mcl{D}) \eqvc \theta(\mcl{D})$. 
\end{proposition}

\begin{proof}
    Assume loop $\mcl{A}$ is of the form: $X \xrightarrow{a_1}X_1\xrightarrow{a_2}...\xrightarrow{a_n}X_n\xrightarrow{a_{n+1}}X$ ($n \ge 0$), with $X \eqvn \theta(x)$. 
	From $x$, by the bisimulation relation $\theta$ between $\mcl{G}$ and $\mcl{H}$, there is a path of the form:
	$x \xrightarrow{a_1} x_{1}\xrightarrow{a_2}...\xrightarrow{a_n}x_{n}\xrightarrow{a_{n+1}} x^{(1)}\xrightarrow{a_1}...
	\xrightarrow{a_{n+1}} x^{(k)}\xrightarrow{a_1} x^{(k)}_{1}\xrightarrow{a_2}...\xrightarrow{a_n}x^{(k)}_{n}\xrightarrow{a_{n+1}}
	x^{(k+1)}\xrightarrow{a_1}...$, 
 where $X_i \eqvn \theta(x^{(k)}_{i})$ for all $k \ge 1$ and $1\le i\le n$, 
 and $X \eqvn \theta(x^{(k)})$ for all $k\ge 1$. 
	Since $\mcl{G}$ is finite, we can always find the node $x^{(m)}$ with the smallest $m > 0$ such that $x^{(m)}\eqvn x^{(j)}$ for some $1\le j < m$. 
	So from $x$ there is a path $\mcl{S}$: $x\xrightarrow{a_1}x_{1}\xrightarrow{a_2}...\xrightarrow{a_n}x_{n}\xrightarrow{a_{n+1}} x^{(1)}\xrightarrow{a_1}...\xrightarrow{a_{n+1}} x^{(j)}$ reaching the loop $\mcl{D}$: $x^{(j)} \xrightarrow{a_1}...\xrightarrow{a_{n+1}} x^{(m)} \eqvn x^{(j)}$. 
 And since $V_\mcl{A} = \theta(V_\mcl{S}\cup V_\mcl{D}) = \theta(V_\mcl{D})$, 
 clearly $\mcl{A} \eqvc \theta(\mcl{S}\unionc \mcl{D}) \eqvc \theta(\mcl{D})$. 
 \qed
\end{proof}

\begin{lemma}
\label{lemma:Image Reflection on LLEE Charts --- critical condition}
    Given the set of images $\mcl{I}_\theta(\mcl{H})$ w.r.t. 
    a bisimulation function $\theta: V_{\hat{\mcl{G}}/\mcl{G}}\to V_\mcl{H}$ from an LLEE chart $\hat{\mcl{G}}/\mcl{G}$ to its bisimulation collapse $\mcl{H}$,  
    the following conditions hold:
    \begin{enumerate}[(1)]
        \item 
        \label{item:lemma:Image Reflection on LLEE Charts --- critical condition - coverage}
        For each loop $\mcl{C}$ of $\mcl{H}$, there exists an image $\mcl{I}$ of $\mcl{I}_\theta(\mcl{H})$ such that $\mcl{C}\subc \mcl{I}$;
        
        \item 
        \label{item:Image Reflection on LLEE Charts --- critical condition - constrained behaviour}
        For each image $\theta(\LC_x)$ of $\mcl{I}_\theta(\mcl{H})$ with a well-structured pre-image $\LC_x$, 
        in $\theta(\LC_x)$
        each loop either passes through node $\theta(x)$ or is in a proper sub-image of $\theta(\LC_x)$. 
        
        \item 
        \label{item:lemma:Image Reflection on LLEE Charts --- critical condition - path}
        For each image $\theta(\LC_x)$ of $\mcl{I}_\theta(\mcl{H})$ with a pre-image $\LC_x$, 
        there exists no transition $U\xrightarrow{}\surd$ such that $U\in V_{\theta(\LC_x)}\setminus\{\theta(x)\}$, and no transition $U\xrightarrow{} W$ satisfying that $U\in V_{\theta(\LC_x)}\setminus\{\theta(x)\}$ and $W\notin V_{\theta(\LC_x)}$. 
        
    \end{enumerate}
\end{lemma} 

Lemma~\ref{lemma:Image Reflection on LLEE Charts --- critical condition} (\ref{item:lemma:Image Reflection on LLEE Charts --- critical condition - coverage}) states the coverage of the loops of $\mcl{H}$ by $\mcl{I}_\theta(\mcl{H})$. 
Lemma~\ref{lemma:Image Reflection on LLEE Charts --- critical condition} (\ref{item:Image Reflection on LLEE Charts --- critical condition - constrained behaviour}) and (\ref{item:lemma:Image Reflection on LLEE Charts --- critical condition - path}) describe the constrained relation between the behaviour of an image and that of any of its well-structured pre-images. 
Lemma~\ref{lemma:Image Reflection on LLEE Charts --- critical condition} (\ref{item:lemma:Image Reflection on LLEE Charts --- critical condition - path}) intuitively means that each path starting from node $\theta(x)$ stays in $\theta(\LC_x)$ before returning to $\theta(x)$. 

\begin{proof}[Proof of Lemma~\ref{lemma:Image Reflection on LLEE Charts --- critical condition}]

(Proof of (1)).  
For each loop $\mcl{C}$ of $\mcl{H}$, by Prop.~\ref{prop:loop correspondence} there is a loop $\mcl{D}$ in $\mcl{G}$ such that $\mcl{C} \eqvc \theta(\mcl{D})$. 
Since $\mcl{G}$ is an LLEE chart, $\mcl{D}$ begins with a loop-entry transition in the form: $x\xrightarrow{}_{[m]}\cdot\xrightarrow{}^* x$, $m > 0$.  
By Def.~\ref{def:Looping-back Chart}, $\mcl{D}$ is a sub-chart of
the looping-back chart $\LC_x$. 
By Prop.~\ref{prop:monotonic of bisimlation function}, $\mcl{C} \eqvc \theta(\mcl{D})\subc \theta(\LC_x)$. 


\ifx
\textit{Base case:} 
$\theta(\LC_x)$ is a minimum image $\mcl{I}$ of $\mcl{I}_\theta(\mcl{H})$.
Assume there is a loop $\mcl{C}$ of $\theta(\LC_x)$ that does not pass $\theta(x)$. 
By Prop.~\ref{prop:loop correspondence}, starting from a node $v$ in $\LC_x$ such that $\theta(v)$ is in $\mcl{C}$, there is a path $\mcl{S}$ reaching a loop $\mcl{D}$ such that $\mcl{C} \eqvc \theta(\mcl{S}\unionc \mcl{D})$ (note that this loop $\mcl{D}$ does not have to be a sub-chart of $\LC_x$). 
Since $\mcl{C}$ does not pass $\theta(x)$,  
chart $\mcl{S}\unionc \mcl{D}$ does not pass $x$. 
So by Prop.~\ref{prop:looping-back chart} (\ref{item: looping-back property - path}), loop $\mcl{D}$ is a sub-chart of $\LC_x$ and by Prop.~\ref{prop:looping-back chart} (\ref{item: looping-back property - loop})
is a sub-chart of a looping-back chart $\LC_k$ of $\LC_x$.
Since $\LC_k\psubc \LC_x$ (by Prop.~\ref{prop:looping-back chart} (\ref{item: looping-back property - looping-back sub-chart})) and $\theta(\LC_x)$ is well structured for $\LC_x$, it must be $\theta(\LC_k)\psubc \theta(\LC_x)$ by Prop.~\ref{prop:well-structured Image}. 
\ifx
By Prop.~\ref{def:well-structured Image} there exists a well-structured image $\theta(\LC_{k'})$ of $\mcl{I}_\theta(\mcl{H})$ such that $\theta(\LC_{k'}) = \theta(\LC_k)$. 
So $\theta(\LC_{k'})\subset \theta(\LC_r)$, 
\fi
This violates that $\theta(\LC_x)$ is a minimum image of $\mcl{I}_\theta(\mcl{H})$. 
\fi

(Proof of (2)).
Let $\mcl{C}$ be a loop of $\theta(\LC_x)$.  
By Prop.~\ref{prop:loop correspondence}, starting from a node $y$ in $\LC_x$ such that $\theta(y)$ is in $\mcl{C}$, there is a path $\mcl{S}$ reaching a loop $\mcl{D}$ such that $\mcl{C} \eqvc \theta(\mcl{S}\unionc \mcl{D})$ (note that this loop $\mcl{D}$ does not have to be a sub-chart of $\LC_x$). 
If $\mcl{C}$ does not pass $\theta(x)$,  
then chart $\mcl{S}\unionc \mcl{D}$ does not pass $x$, and $y\in V_{\LC_x}\setminus\{x\}$.  
By Prop.~\ref{prop:looping-back chart} (\ref{item: looping-back property - path}) and (\ref{item: looping-back property - termination}), since $y\in V_{\LC_x}\setminus\{x\}$, loop $\mcl{D}$ stays in $\LC_x$ and is thus a sub-chart of $\LC_x$. 
By that $\mcl{G}$ is an LLEE chart, 
loop $\mcl{D}$ starts from a loop-entry transition $k\xrightarrow{}_{[m]} \cdot$ with $m > 0$.
By Def.~\ref{def:Looping-back Chart}, it is not hard to see that $\mcl{D}$ must be a sub-chart of the looping-back chart $\LC_k$. 
So $\mcl{D}\subc \LC_k\subc \LC_x$. 
Since $k\not\!\!\eqvn x$, by Prop.~\ref{prop:looping-back chart} (\ref{item: looping-back property - looping-back sub-chart}), 
$\LC_k\psubc \LC_x$. 
According to Prop.~\ref{prop:well-structured Image} and that $\LC_x$ is well structured, $\theta(\LC_k)\psubc \theta(\LC_x)$. 
\ifx
And by Prop.~\ref{prop:looping-back chart} (\ref{item: looping-back property - loop}), 
$\mcl{D}$ is a sub-chart of the looping-back chart $\LC_k$ of $\LC_x$.
Since $\LC_k\psubc \LC_x$ (by Prop.~\ref{prop:looping-back chart} (\ref{item: looping-back property - looping-back sub-chart})) and $\theta(\LC_x)$ is well structured for $\LC_x$, it must be $\theta(\LC_k)\psubc \theta(\LC_x)$ by Prop.~\ref{prop:well-structured Image}. 
\fi
\ifx
Similarly as in the base case, it is easy to see that chart $\mcl{S}\unionc\mcl{D}$ does not pass $x$ and 
so $\mcl{D}$ is in some looping-back sub-chart $\LC_k$ of $\LC_x$ (by Prop.~\ref{prop:looping-back chart} (\ref{item: looping-back property - path}) (\ref{item: looping-back property - loop})), moreover, with $\LC_k\psubc \LC_x$ (by Prop.~\ref{prop:looping-back chart} (\ref{item: looping-back property - looping-back sub-chart})). 
Therefore, $\mcl{C} \eqvc \theta(\mcl{S}\unionc \mcl{D}) \eqvc \theta(\mcl{D}) \subc\theta(\LC_k)\subc \theta(\LC_x)$ (by Prop.~\ref{prop:loop correspondence} and~\ref{prop:monotonic of bisimlation function}).  
By that $\theta(\LC_x)$ is well structured for $\LC_x$ and Prop.~\ref{prop:well-structured Image}, 
$\theta(\LC_k)\psubc \theta(\LC_x)$.  
\fi
\ifx
And by Prop.~\ref{def:well-structured Image} there exists a well-structured image $\theta(\LC_{k'})$ of $\mcl{I}_\theta(\mcl{H})$ such that $\theta(\LC_{k'}) = \theta(\LC_k)$. 
So $\mcl{C}\subseteq \theta(\LC_{k'})\subset \theta(\LC_s)$, violating the assumption about $\mcl{C}$. 
\fi
So $\mcl{C}\eqvc \theta(\mcl{S}\unionc \mcl{D})\eqvc \theta(\mcl{D})$ (by Prop.~\ref{prop:loop correspondence}) 
$\subc \theta(\LC_k)$ (by Prop.~\ref{prop:monotonic of bisimlation function})
$\psubc \theta(\LC_x)$. 
In other words, $\mcl{C}$ is in the proper sub-image $\theta(\LC_k)$ of $\theta(\LC_x)$. 

(Proof of (3)).
We prove by contradiction. 
Assume there exists a transition $U\xrightarrow{a}\surd$ such that 
$U\in V_{\theta(\LC_x)}\setminus\{\theta(x)\}$, or a transition $U\xrightarrow{a} W$ such that 
$U\in V_{\theta(\LC_x)}\setminus\{\theta(x)\}$ but $W\notin V_{\theta(\LC_x)}$. 
Let $u$ be the node in $\LC_x$ such that $\theta(u) \eqvn U$. 
Since $U\not\equiv \theta(x)$, $u\not\equiv x$. So $u\in V_{\LC_x}\setminus\{x\}$. 
For the case of $U\xrightarrow{a}\surd$, 
by Prop.~\ref{prop:looping-back chart} (\ref{item: looping-back property - termination}), $u\not\xrightarrow{a} \surd$. So this causes contradiction since $u\sim U$.
For the case of $U\xrightarrow{a} W$, 
there is a transition $u\xrightarrow{a} w$ such that $\theta(w) \eqvn W$ and $w\notin V_{\LC_x}$ (otherwise $W\in V_{\theta(\LC_x)}$). 
But the fact that $u\in V_{\LC_x}\setminus\{x\}$ and $w\notin V_{\LC_x}$ violates Prop.~\ref{prop:looping-back chart} (\ref{item: looping-back property - path}). 
Therefore, transition $U\xrightarrow{a} W$ does not exist. 
\qed

\ifx
\textit{Base case:} 
Consider a minimum image $\theta(\LC_r)$ of $\mcl{I}_\theta(\mcl{H})$. 
If there exists a transition $U\xrightarrow{a} V$ such that 
$U\in \theta^{\ip}(\LC_r)$ but $V\notin \theta(\LC_r)$ (which violates (3)), 
then there is a node $u$ in $\LC_r$ such that $\theta(u) = U$. 
So there is a transition $u\xrightarrow{a} v$ such that $\theta(v) = V$ and $v\notin \LC_r$. 
Note that since $u\not\equiv r$, it violates that $\mcl{G}$ is an LLEE chart ((L2) does not hold for $\LC_r$). 
Therefore, such a transition $U\xrightarrow{a} V$ does not exist. 

\textit{Step case:} 
For an arbitrary image $\theta(\LC_s))$ of $\mcl{I}_\theta(\mcl{H})$, 
suppose there is a transition $U\xrightarrow{a} V$ in $\theta(\LC_s)$ where $U\in \theta^\ip(\LC_s)$ but $V\notin \theta(\LC_s)$.  
By inductive hypothesis, $U\notin \theta^\ip(\LC_r)$ for  any sub-image $\theta(\LC_r)$ of image $\theta(\LC_s)$.  
Similar to the base case above, we can prove that transition $U\xrightarrow{a} V$ causes contradiction to the fact that $\LC_s$ is a loop sub-chart of $\mcl{G}$. 
\fi
\end{proof}

\begin{theorem}[Stage (II')]
\label{theo:restoration for LLEE charts}
    For every bisimulation function $\theta: V_{\hat{\mcl{G}}/\mcl{G}}\to V_\mcl{H}$ from an LLEE chart $\hat{\mcl{G}}/\mcl{G}$ to a bisimulation collapse $\mcl{H}$, $\mcl{H}$ is an LEE chart. 
\end{theorem}

\begin{proof}
\ifx
We build a structure $(\mcl{I}_\theta(\mcl{H}), \eo)$ according to Def.~\ref{def:Elimination Order}. 
$\eo$ can be obtained by stipulating an order for every two images $\mcl{I}_1$ and $\mcl{I}_2$ such that $\mcl{I}_1\cap \mcl{I}_2\neq \emptyset$ in an inductive way. 
Let $\reg(\mcl{H})$ be the set of well-structured regions on $(\mcl{I}_\theta(\mcl{H}), \eo)$. 
\fi

By the definition of LEE chart in Sect.~\ref{section:LEE-LLEE Charts}, 
we show that there is an elimination process on $\mcl{H}$ 
by which $\mcl{H}$ can be transformed into a chart $\mcl{H}'$ without infinite paths. 
It is sufficient to prove that each image of $\mcl{I}_\theta(\mcl{H})$ is an LEE sub-chart of $\mcl{H}$ 
by induction on the sub-image relations between images. 
The elimination process on $\mcl{H}$ can then be naturally induced by this proof process. 



\textit{Base case:}	
Consider a minimum image
$\mcl{I}$ of $\mcl{I}_\theta(\mcl{H})$ (in which there are no proper sub-images). 
By Prop.~\ref{prop:well-structured Image}, let $\mcl{I}\eqvc \theta(\LC_r)$ for some $\LC_r$ of its well-structured pre-images.
If $\theta(\LC_r)$ has at least one loop, 
we show that it is a loop sub-chart of $\mcl{H}$ (so is also an LEE sub-chart of $\mcl{H}$). 
(L1) is obvious since $\theta(\LC_r)$ has a loop. 
By Lemma~\ref{lemma:Image Reflection on LLEE Charts --- critical condition} (2) and that $\theta(\LC_r)$ is a minimum image, each loop in $\theta(\LC_r)$ must pass through $\theta(r)$. 
By Lemma~\ref{lemma:Image Reflection on LLEE Charts --- critical condition} (3), 
each path in $\theta(\LC_r)$ starting from $\theta(r)$ cannot leave $\theta(\LC_r)$ before returning to $\theta(r)$. 
So, for each infinite path in $\theta(\LC_r)$ starting from $\theta(r)$, 
it can only return to $\theta(r)$ after a finite number of transitions. 
That is, (L2) holds.  
(L3) is direct from Prop.~\ref{prop:looping-back chart} (\ref{item: looping-back property - termination}) satisfied by $\LC_r$ and the bisimulation relation $\theta$ itself. 

\textit{Step case:}
Consider an arbitrary image $\mcl{I}$ of $\mcl{I}_\theta(\mcl{H})$. 
By Prop.~\ref{prop:well-structured Image}, let $\mcl{I}\eqvc \theta(\LC_s)$ for some well-structured pre-image $\LC_s$. 
We show that $\theta(\LC_s)$ is an LEE sub-chart. 
By inductive hypothesis, each maximum proper sub-image $\mcl{I}'$ of $\theta(\LC_s)$ is an LEE sub-chart.
So there exists an elimination process to eliminate each $\mcl{I}'$.  
After eliminating all these sub-images, assume $\theta(\LC_s)$ is transformed into a chart $(\theta(\LC_s))'$. 
By Lemma~\ref{lemma:Image Reflection on LLEE Charts --- critical condition} (2), all loops remained in $(\theta(\LC_s))'$ must pass through $\theta(s)$. 
If $(\theta(\LC_s))'$ still has at least one loop, 
then similarly as in the base case above, we show that $(\theta(\LC_s))'$ is a loop sub-chart: 

(L1) is obvious because $(\theta(\LC_s))'$ has a loop. 
By Lemma~\ref{lemma:Image Reflection on LLEE Charts --- critical condition} (3), each path in $\theta(\LC_s)$ cannot leave $\theta(\LC_s)$ before returning to $\theta(s)$. 
So after the eliminations, each path in the remnant $(\theta(\LC_s))'$ still cannot leave $(\theta(\LC_s))'$ before returning to $\theta(s)$. 
By this and the fact that all loops in $(\theta(\LC_s))'$ pass through $\theta(s)$ as obtained above, 
it is easy to see that (L2) holds. 
(L3) is direct from Prop.~\ref{prop:looping-back chart} (\ref{item: looping-back property - termination}) satisfied by $\LC_s$ and the bisimulation relation $\theta$ itself. 

By that $(\theta(\LC_s))'$ is a loop sub-chart, it then can be eliminated by an elimination process.  
So the whole image $\theta(\LC_s)$ can be eliminated and is thus an LEE sub-chart. 

\textit{Conclusion:}
Since each image is an LEE sub-chart of $\mcl{H}$, 
as indicated in the inductive reasoning above, 
we obtain an elimination process that eliminates all images in $\mcl{I}_\theta(\mcl{H})$ in an inside-out manner from the minimum sub-images according to the sub-image relations. 
Let $\mcl{H}'$ be the chart $\mcl{H}$ after the elimination process. 
By Lemma~\ref{lemma:Image Reflection on LLEE Charts --- critical condition} (1), there exist no loops left in $\mcl{H}'$. 
So $\mcl{H}'$ is a chart without infinite paths. 
\qed
\end{proof}

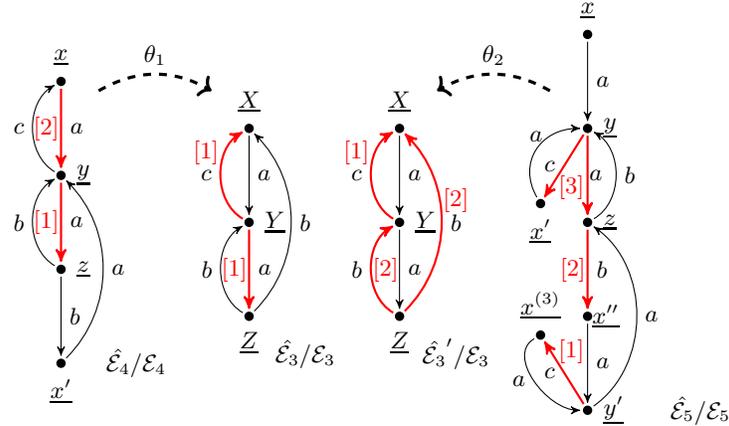
\begin{figure}[tbhp]
    \centering
    \noindent\makebox[\textwidth]{ 
    \begin{tabular}{c}
    \begin{tikzpicture}[->,>=stealth', node distance=1.7cm]

     \coordinate (x) at (0,0);
     \coordinate (y) at ($(x) + (0,-1.25cm)$);
     \coordinate (z) at ($(y) + (0,-1.25cm)$);

     \coordinate (X1) at ($(x) + (-2.5cm,0.625cm)$);
     \coordinate (Y1) at ($(X1) + (0,-1.25cm)$);
     \coordinate (Z1) at ($(Y1) + (0,-1.25cm)$);
     \coordinate (X1i) at ($(Z1) + (0,-1.25cm)$);

     \coordinate (xi) at ($(x) + (2cm,0)$);
     \coordinate (yi) at ($(xi) + (0,-1.25cm)$);
     \coordinate (zi) at ($(yi) + (0,-1.25cm)$);
     
     \coordinate (X2) at ($(xi) + (2.5cm,1.25cm)$);
     \coordinate (Y2) at ($(X2) + (0,-1.25cm)$);
     \coordinate (X2i) at ($(Y2) + (-0.625cm, -1cm)$);
     \coordinate (Z2) at ($(Y2) + (0,-1.25cm)$);
     \coordinate (X2ii) at ($(Z2) + (0,-1.25cm)$);
     \coordinate (Y2i) at ($(X2ii) + (0,-1.25cm)$);
     \coordinate (X2iii) at ($(Y2i) + (-0.625cm,1cm)$);

    \node[circle,fill,inner sep=1.25pt, outer sep=1pt, label = north:$\underline{X}$] at (x) (nx) {};
    \node[circle,fill,inner sep=1.25pt, outer sep=1pt, label = east:$\underline{Y}$] at (y) (ny) {};
    \node[circle,fill,inner sep=1.25pt, outer sep=1pt, label = south:$\underline{Z}$] at (z) (nz) {};
    \node[circle,fill,inner sep=1.25pt, outer sep=1pt, label = north:$\underline{x}$] at (X1) (nX1) {};
    \node[circle,fill,inner sep=1.25pt, outer sep=1pt, label = east:$\underline{y}$] at (Y1) (nY1) {};
    \node[circle,fill,inner sep=1.25pt, outer sep=1pt, label = east:$\underline{z}$] at (Z1) (nZ1) {};
    \node[circle,fill,inner sep=1.25pt, outer sep=1pt, label = south:$\underline{x'}$] at (X1i) (nX1i) {};

    \node[circle,fill,inner sep=1.25pt, outer sep=1pt, label = north:$\underline{X}$] at (xi) (nxi) {};
    \node[circle,fill,inner sep=1.25pt, outer sep=1pt, label = east:$\underline{Y}$] at (yi) (nyi) {};
    \node[circle,fill,inner sep=1.25pt, outer sep=1pt, label = south:$\underline{Z}$] at (zi) (nzi) {};
    
    \node[circle,fill,inner sep=1.25pt, outer sep=1pt, label = north:$\underline{x}$] at (X2) (nX2) {};
    \node[circle,fill,inner sep=1.25pt, outer sep=1pt, label = east:$\underline{y}$] at (Y2) (nY2) {}; 
    \node[circle,fill,inner sep=1.25pt, outer sep=1pt, label = south:$\underline{x'}$] at (X2i) (nX2i) {};
    \node[circle,fill,inner sep=1.25pt, outer sep=1pt, label = east:$\underline{z}$] at (Z2) (nZ2) {};
    \node[circle,fill,inner sep=1.25pt, outer sep=1pt, label = {[label distance=-0.15cm]east:$\underline{x''}$}] at (X2ii) (nX2ii) {};
    \node[circle,fill,inner sep=1.25pt, outer sep=1pt, label = east:$\underline{y'}$] at (Y2i) (nY2i) {};
    \node[circle,fill,inner sep=1.25pt, outer sep=1pt, label = north:$\underline{x^{(3)}}$] at (X2iii) (nX2iii) {};
    
     \path
     (nx) edge [] node[right, yshift=0cm,black] {$a$} (ny)
     (ny) edge [red,thick] node[left,xshift=0.1cm] {\small{$[1]$}} node[right,black] {$a$} (nz)
     (ny) edge [out=150,in=-150,red,thick] node[left,yshift=0.3cm,xshift=0.1cm] {\small{$[1]$}} node[left,black] {$c$} (nx)
     (nz) edge [out=150,in=-150] node[left] {$b$} (ny)
     (nz) edge [out=45,in=-45] node[right] {$b$} (nx)

     (nX1) edge [red,thick] node[left,xshift=0.1cm] {\small{$[2]$}} node[right,black] {$a$} (nY1)
     (nY1) edge [red,thick] node[left,xshift=0.1cm] {\small{$[1]$}} node[right,black] {$a$} (nZ1)
     (nY1) edge [out=150,in=-150] node[left] {$c$} (nX1)
     (nZ1) edge [out=150,in=-150] node[left] {$b$} (nY1)
     (nZ1) edge [] node[right] {$b$} (nX1i)
     (nX1i) edge [out=45,in=-45] node[right] {$a$} (nY1)

     (nxi) edge [] node[right, yshift=0cm,black] {$a$} (nyi)
     (nyi) edge [] node[right,black] {$a$} (nzi)
     (nyi) edge [out=150,in=-150,red,thick] node[left,yshift=0.3cm,xshift=0.1cm] {\small{$[1]$}} node[left,black] {$c$} (nxi)
     (nzi) edge [out=150,in=-150,red,thick] node[right,xshift=-0.1cm] {\small{$[2]$}} node[left, black] {$b$} (nyi)
     (nzi) edge [out=45,in=-45,red,thick] node[right,yshift=0.3cm,xshift=-0.1cm] {\small{$[2]$}} node[right,black] {$b$} (nxi)
     
     (nX2) edge node[right] {$a$} (nY2)
     (nY2) edge[red,thick] node[below,xshift=0.1cm]{\small{$[3]$}} node[left,black] {$c$} (nX2i)
     (nY2) edge[red,thick] node[right,xshift=-0.1cm, black] {$a$} (nZ2)
     (nX2i) edge[out=120,in=180,looseness=1.25] node[above] {$a$} (nY2)
     (nZ2) edge[out=30,in=-30] node[right] {$b$} (nY2)
     (nZ2) edge[red,thick] node[left,xshift=0.1cm] {\small{$[2]$}} node[right,black] {$b$} (nX2ii)
     (nX2ii) edge[] node[right] {$a$} (nY2i)
     (nY2i) edge[out=35,in=-35] node[right] {$a$} (nZ2)
     (nY2i) edge[red,thick] node[above,xshift=0.1cm] {\small{$[1]$}} node[left, black] {$c$} (nX2iii)
     (nX2iii) edge[out=210,in=180] node[left] {$a$} (nY2i)
    ;

    \coordinate (P1) at ($(x) + (-0.5cm,0.5cm)$); 
    \coordinate (P2) at ($(P1) + (-1.5cm,0)$); 
    \coordinate (P3) at ($(xi) + (0.5cm,0.5cm)$);
    \coordinate (P4) at ($(P3) + (1.5cm,0)$);

    \draw[-To,dashed,very thick,out=30,in=150] (P2) to node[above] {$\theta_1$} (P1); 
    \draw[-To,dashed,very thick,out=150,in=30] (P4) to node[above] {$\theta_2$} (P3); 

    \coordinate (EII1) at ($(z) + (0.75cm,-0.5cm)$);
    \coordinate (EII2) at ($(zi) + (0.75cm,-0.5cm)$);
    \coordinate (EIIi) at ($(X1i) + (1cm,0cm)$);
    \coordinate (EIIii) at ($(Y2i) + (1.5cm,0cm)$);

    \node[draw=none] at (EII1) (nEII1) {$\hat{\EII}/\EII$};
    \node[draw=none] at (EII2) (nEII2) {$\hat{\EII}'/\EII$};
    \node[draw=none] at (EIIi) (nEIIi) {$\hat{\EIIi}/\EIIi$};
    \node[draw=none] at (EIIii) (nEIIii) {$\hat{\EIIii}/\EIIii$};
    
    \end{tikzpicture}
    \\
    \end{tabular}
    } 
        \caption{An example of different image reflections by a same chart}
        \label{figure:An example of image reflection on LLEE charts}
\end{figure}

\ifx
Readers may wonder
the above proof does not reveal a situation when two images, say $\theta(\LC_x)$ and $\theta(\LC_y)$, of $\mcl{I}_\theta(\mcl{H})$ overlaps each other. 
Our answer is, it does not really matter. 
Because during the elimination process described above, 
even if $\theta(\LC_x)$ and $\theta(\LC_y)$ are overlapped, their eliminations are irrelevant to each other. 
What only matters is, as indicated in the proof above, that both $\theta(\LC_x)$ and $\theta(\LC_y)$, after 
eliminating their maximum sub-images, can become two loop charts. 
This hierarchy structure of well-constructed images indicated by sub-image relations is the key for our proof. 
\fi

In the above proof, 
for each image of $\mcl{I}_\theta(\mcl{H})$, 
only one well-structured pre-image is needed to guarantee that 
the image is an LEE sub-chart. 
The hierarchy structure of images indicated by the sub-image relations between them is the key to determine the LEE elimination process on $\mcl{H}$. 

\ifx
Note that two images that are not sub-images to each other may overlap each other. However, this does not affect each of their eliminations once an elimination process is determined. 
\fi

For example, consider the chart $\CI$ of Fig.~\ref{figure:An example of images}. 
By the elimination process of its images we have discussed earlier at the beginning of this section,  we obtained its LEE witness $\hat{\CI}$ as also shown on the left side in Fig.~\ref{figure:An example of images}. 
In this example, image $\CI$ reflects the LLEE structure 
of its well-structured pre-image $\gra(\{x, z'', k, y\})_{x}$;
images $\gra(\{Z, X\})$ and $\gra(\{Z, Y\})$ reflect the LLEE structures of their well-structured pre-images $\gra(\{z', x'\})_{z'}$ and $\gra(\{z'', y\})_{z''}$ respectively. 
\ifx
It does not matter that images $\gra(\{Z, X\})$ and $\gra(\{Z, Y\})$ are overlapped. After eliminating $\gra(\{Z, Y\})$ first, the remnant chart 
$\la \{Z, X\}, \{a_1,a_2,a_3,a_4,b_1,d_1,d_2\}, \{Z\xrightarrow{a_2} X, X\xrightarrow{b_1} Z\}, \surd\}$ of $\gra(\{Z, X\})$ (without transition $Z\xrightarrow{a_1} Z$) is still a loop sub-chart 
and can be eliminated further. 
\fi

In fact, from Theorem~\ref{theo:restoration for LLEE charts}, 
a bisimulation collapse $\mcl{H}$ reflects the LLEE structure of 
some parts of \emph{any} LLEE chart that has $\mcl{H}$ as its bisimulation collapse. 
This result coincides with the nature that an LLEE chart can have multiple LEE/LLEE witnesses.

Consider another example firstly proposed in~\cite{Grabmayer20} as shown in Fig.~\ref{figure:An example of image reflection on LLEE charts}. 
Given two bisimulation functions $\theta_1: V_{\hat{\EIIi}/\EIIi}\to V_{\EII}$ and 
$\theta_2: V_{\hat{\EIIii}/\EIIii}\to V_{\EII}$ from LLEE charts $\EIIi$ and $\EIIii$ to their bisimulation collapse $\EII$, 
$\EII$ reflects some parts of the LLEE structures of both $\EIIi$ and $\EIIii$:
$\EII$ with the left witness $\hat{\EII}$ reflects the LLEE structure of the looping-back chart $\gra(\{y, z, x'\})_{y}$ of $\EIIi$, whereas $\EII$ with the right witness $\hat{\EII}'$ 
reflects the LLEE structure of the looping-back chart $\gra(\{z, x'', y', x^{(3)}\})_{z}$ of $\EIIii$.  
In our method, 
the different witnesses of  $\hat{\EII}$ and $\hat{\EII}'$ come from different structures of images on $\EII$ through two different bisimulation functions $\theta_1$ and $\theta_2$.

Now we illustrate that Theorem~\ref{theo:restoration for LLEE charts} 
is a result no weaker than the one obtained in stage (II). 


As stated in Sect.~\ref{section:LEE-LLEE Charts}, not all LEE witnesses are LLEE witnesses. 
The following property was firstly obtained in~\cite{Grabmayer2022-2v}. 

\begin{proposition}[\cite{Grabmayer2022-2v,Grabmayer2022}]
\label{prop:From LEE Witnesses to LLEE Witnesses}
Each LEE chart is also an LLEE chart. 
\end{proposition}

\ifx
This property was mentioned in~\cite{???} but without an explicit proof. 
We give our version of its proof in this paper (see Appendix~\ref{section:Other Proofs}). 
To prove Prop.~\ref{prop:From LEE Witnesses to LLEE Witnesses}, we show each LEE witness $\hat{\mcl{G}}$ of $\mcl{G}$ can be transformed into an LLEE witness $\hat{\mcl{G}}'$ so that $\hat{\mcl{G}}'/\mcl{G}$ satisfies LLEE property. 
\fi

Though not giving an explicit proof, \cite{Grabmayer2022-2v} pointed out that a proof can be obtained by showing that 
each LEE witness $\hat{\mcl{G}}$ of $\mcl{G}$ can be transformed into an LLEE witness $\hat{\mcl{G}}'$ so that $\hat{\mcl{G}}'/\mcl{G}$ satisfies the LLEE property. 
Here, we give our proof as below. 

\begin{proof}[Proof of Prop.~\ref{prop:From LEE Witnesses to LLEE Witnesses}]

    Assume $\hat{\mcl{G}}/\mcl{G}$ is an LEE chart with $\hat{\mcl{G}}$ an LEE witness. 
    We obtain an LLEE witness $\hat{\mcl{G}}'$ from $\hat{\mcl{G}}$ by induction on the LEE elimination process according to $\hat{\mcl{G}}$.  
    In the following proof, 
    without loss of generality, assume that every loop-entry transition has an unique order number, starting from the smallest number $1$. 
    We denote by $\mcl{D}(X\xrightarrow{} Y)$ a loop starting from transition 
    $X\xrightarrow{} Y$.  
    We simply use $\la X, \E\ra$ to denote a $\la X, \E\ra$-generated chart of $\mcl{G}$.  

    \textit{Base case:} 
    Consider the $\la R_1, \{R_1\xrightarrow{}_{[1]}\cdot\}\ra$-generated chart with $R_1 \xrightarrow{}_{[1]}\cdot$ the smallest loop-entry transition. 
    By LEE property, $\la R_1, \{R_1\xrightarrow{}_{[1]}\cdot\}\ra$ is a loop sub-chart.
    For each loop-entry transition $X\xrightarrow{}_{[k]} \cdot$ from the body of chart $\la R_1, \{R_1\xrightarrow{}_{[1]}\cdot\}\ra$ with $k > 1$ (thus violating the LLEE property),
    we turn $X\xrightarrow{}_{[k]} \cdot$ into a body transition: $X\xrightarrow{}_{[0]} \cdot$. 
    For each loop $\mcl{D}(X\xrightarrow{}_{[0]}\cdot)$ that becomes a loop without loop-entry transitions because of the transforming from $X\xrightarrow{}_{[k]}\cdot $ to $X\xrightarrow{}_{[0]}\cdot $, 
    since $X\xrightarrow{}_{[0]}\cdot$ is from the body of $\la R_1, \{R_1\xrightarrow{}_{[1]}\cdot\}\ra$, 
    by that $\la R_1, \{R_1\xrightarrow{}_{[1]}\cdot\}\ra$ is a loop sub-chart and (L2), $\mcl{D}(X\xrightarrow{}_{[0]}\cdot)$ must pass through $R_1$. 
    Let $\mcl{D}(X\xrightarrow{}_{[0]}\cdot)$ be the form of: $X\xrightarrow{}_{[0]} \cdot \xrightarrow{}^* R_1\xrightarrow{}Z_1 \xrightarrow{}^* X$. 
    We then turn transition $R_1\xrightarrow{} Z_1$ into a loop-entry transition $R_1\xrightarrow{}_{[k]}Z_1 $. 
    It is not hard to see that after the loop-entry-transition-switching process described above for each $X\xrightarrow{}_{[k]} \cdot$ with $k > 1$, 
    all transitions $\cdot \xrightarrow{}_{[k]}\cdot $ from the body of $\la R_1, \{R_1\xrightarrow{}_{[1]}\cdot\}\ra$ are body transitions with $k = 0 < 1$. 

    \ifx
    Consider the loop sub-chart $\mcl{C}{[1]}$ whose largest loop-entry transition is the smallest loop-entry transition $R_1 \xrightarrow{}_{[1]}\cdot$. 
    By LEE property, $\mcl{C}{[1]}$ is a loop chart. 
    For each loop-entry transition $X\xrightarrow{}_{[k]} \cdot$ from the body of $\mcl{C}{[1]}$ with $k > 1$ (violating the LLEE property),
    we turns $X\xrightarrow{}_{[k]} \cdot$ into a body transition: $X\xrightarrow{}_{[0]} \cdot$. 
    And for each loop $\mcl{D}(X\xrightarrow{}_{[0]}\cdot)$ that becomes a loop without loop-entry transitions (because of the transforming of $X\xrightarrow{}_{[k]}\cdot $), 
    since $X\xrightarrow{}_{[0]}\cdot$ is from the body of $\mcl{C}{[1]}$, 
    by that $\mcl{C}{[1]}$ is a loop chart and (L2), $\mcl{D}(X\xrightarrow{}_{[0]}\cdot)$ must pass through $R_1$.
    Let $\mcl{D}(X\xrightarrow{}_{[0]}\cdot)$ be the form of: $X\xrightarrow{}_{[0]} \cdot \xrightarrow{}^* R_1\xrightarrow{}Z_1 \xrightarrow{}^* X$. 
    We turn transition $R_1\xrightarrow{} Z_1$ into a loop-entry transition $R_1\xrightarrow{}_{[k]}Z_1 $.
    It is easy to see that after the loop-entry-transition switching described above for each $X\xrightarrow{}_{[k]} \cdot$, 
    all transitions $\cdot \xrightarrow{}_{[k]}\cdot $ from the body of $\mcl{C}{[1]}$ are body transitions with $k = 0$. 
    \fi
    
    \textit{Step case:} 
    Consider a $\la R_n, \{R_n\xrightarrow{}_{[n]}\cdot \}\ra$-generated chart starting from a loop-entry transition $R_n\xrightarrow{}_{[n]}\cdot $ with $n > 1$.
    By induction hypothesis, in any $\la R_l, \{R_l\xrightarrow{}_{[l]}\cdot \}\ra$-generated chart with $l < n$, 
    any transition 
    $\cdot \xrightarrow{}_{[o]}\cdot$ from the body of $\la R_l, \{R_l\xrightarrow{}_{[l]}\cdot \}\ra$ satisfies that $o < l$. 
    For each loop-entry transition  $X\xrightarrow{}_{[k]} \cdot$ from the body of $\la R_n, \{R_n\xrightarrow{}_{[n]}\cdot \}\ra$ satisfying $k > n$,  
    we turn $X\xrightarrow{}_{[k]} \cdot$ into a body transition: $X\xrightarrow{}_{[0]} \cdot$.
    For each loop $\mcl{D}(X\xrightarrow{}_{[0]}\cdot)$ that becomes a loop without loop-entry transitions because of the transforming from $X\xrightarrow{}_{[k]}\cdot$ to $X\xrightarrow{}_{[0]}\cdot$, 
    we see that loop $\mcl{D}(X\xrightarrow{}_{[0]}\cdot)$ must pass through $R_n$. 
    Because by the LEE property, after eliminating all loop-entry transitions 
    with an order number $l < n$, 
    chart $\la R_n, \{R_n\xrightarrow{}_{[n]}\cdot \}\ra$ becomes a loop sub-chart, denoted by $\la R_n, \{R_n\xrightarrow{}_{[n]}\cdot \}\ra'$.  
    Since $X\xrightarrow{}_{[k]} \cdot$ is not from the body of any chart $\la R_l, \{R_l\xrightarrow{}_{[l]}\cdot \}\ra$ with $l < n$ (by inductive hypothesis above), 
    loop $\mcl{D}(X\xrightarrow{}_{[0]}\cdot)$ is still not eliminated in $\la R_n, \{R_n\xrightarrow{}_{[n]}\cdot \}\ra'$. 
    So by (L2), $\mcl{D}(X\xrightarrow{}_{[0]}\cdot)$ passes through $R_n$. 
    Let $\mcl{D}(X\xrightarrow{}_{[0]}\cdot)$ be of the form:
    $X\xrightarrow{}_{[0]}\cdot \xrightarrow{}^* R_n\xrightarrow{} Z_n \xrightarrow{}^* X$, 
    then we turn the body transition $R_n\xrightarrow{} Z_n$ into 
    a loop-entry transition $R_n\xrightarrow{}_{[k]} Z_n$. 
    After the loop-entry-transition-switching process described above for each $X\xrightarrow{}_{[k]} \cdot$, we see that all
    transitions $\cdot \xrightarrow{}_{[k]}\cdot $ from the body of $\la R_n, \{R_n\xrightarrow{}_{[n]}\cdot \}\ra$ satisfy that $k < n$.

    \ifx
    Consider an arbitrary loop sub-chart $\mcl{C}{[n]}$ whose largest 
    loop-entry transition is $R_n\xrightarrow{}_{[n]}\cdot $ with $n > 1$. 
    By induction hypothesis, in any loop sub-chart $\mcl{C}{[l]}$ with $l < n$, 
    all transitions 
    $\cdot \xrightarrow{}_{[k]}\cdot$ from the body of $\mcl{C}{[l]}$ satisfy that $k < l$. 
    For each loop-entry transition  $X\xrightarrow{}_{[k]} \cdot$ from the body of $\mcl{C}{[n]}$ with $k > n$,  
    we turn $X\xrightarrow{}_{[k]} \cdot$ into a body transition: $X\xrightarrow{}_{[0]} \cdot$.
    And for each loop $\mcl{D}(X\xrightarrow{}_{[0]}\cdot)$ that becomes without loop-entry transitions (because of the transforming of $X\xrightarrow{}_{[k]}\cdot$), 
    we see that $\mcl{D}(X\xrightarrow{}_{[0]}\cdot)$ must pass through $R_n$. 
    Because by the LEE property, after eliminating all $\mcl{C}{[l]}$ with $l < n$, 
    $\mcl{C}{[n]}$ becomes a loop chart. 
    And since $X\xrightarrow{}_{[k]} \cdot$ is not from the body of $\mcl{C}{[l]}$ (by inductive hypothesis), 
    loop $\mcl{D}(X\xrightarrow{}_{[0]}\cdot)$ is not eliminated. 
    So by (L2), $\mcl{D}(X\xrightarrow{}_{[0]}\cdot)$ passes through $R_n$. 
    Let $\mcl{D}(X\xrightarrow{}_{[0]}\cdot)$ be of the form:
    $X\xrightarrow{}_{[0]}\cdot \xrightarrow{}^* R_n\xrightarrow{} Z_n \xrightarrow{}^* X$, 
    then we turn the body transition $R_n\xrightarrow{} Z_n$ into 
    a loop-entry transition $R_n\xrightarrow{}_{[k]} Z_n$. 
    After the loop-entry-transition switching for each $X\xrightarrow{}_{[k]} \cdot$ as above, it is easy to see that all
    transitions $\cdot \xrightarrow{}_{[k]}\cdot $ from the body of $\mcl{C}{[n]}$ satisfy that $k < n$.
    \fi
    
    By the LLEE property, the witness $\hat{\mcl{G}}'$ obtained after 
    the transformations above is the LLEE witness as required. 
    \qed
\end{proof}

As an example, see the LEE chart $\hat{\CI}'/\CI$ with witness $\hat{\CI}'$ in Fig.~\ref{figure:An example of a non-LLEE LEE witness}, which is an LLEE witness. 
From the LEE witness $\hat{\CI}$, 
we first turn the loop-entry transition $X\xrightarrow{b_1}_{[3]} Z$ from the body of the $\la Z, \{Z\xrightarrow{a_2}_{[2]} X\}\ra$-generated chart that violates the LLEE property into a body transition $X\xrightarrow{b_1}_{[0]} Z$. 
This causes loop $X\xrightarrow{b_1} Z\xrightarrow{a_4} K\xrightarrow{d_2} X$ become a loop without loop-entry transitions. 
Then we turn transition $Z\xrightarrow{a_4} K$ into a loop-entry transition 
$Z\xrightarrow{a_4}_{[3]} K$. 
By this loop-entry-transition-switching process we transform $\hat{\CI}$ into an LLEE witness $\hat{\CI}'$. 


We obtain the same result as stage (II) as follows. 
It is straightforward by Prop.~\ref{prop:From LEE Witnesses to LLEE Witnesses} and Theorem~\ref{theo:restoration for LLEE charts}. 

\begin{corollary}[\cite{Grabmayer20}]
    For every bisimulation function $\theta: V_{\hat{\mcl{G}}/\mcl{G}}\to V_\mcl{H}$ from an LLEE chart $\hat{\mcl{G}}/\mcl{G}$ to a bisimulation collapse $\mcl{H}$, $\mcl{H}$ is an LLEE chart. 
\end{corollary}

\section{Conclusion and Analysis}
\label{section:Conclusion and Future Work}

In this paper, we analyze a phenomenon called ``image reflection'' on a type of characterization graphs --- LLEE charts, which leads us to a new approach for proving the completeness of \MilIfree\ for 1-free regular expressions modulo bisimulation equivalence. 
The novelty of our method lies in stage (II) 
for showing that the bisimulation collapse of an LLEE chart is an LLEE chart using the minimization strategy. 
We propose a different approach by utilizing 
the image reflection from the bisimulation collapse $\mcl{H}$ of an LLEE chart $\mcl{G}$ back to chart $\mcl{G}$, through a bisimulation function $f: V_\mcl{G}\to V_\mcl{H}$ from $\mcl{G}$ to $\mcl{H}$. 
\ifx
a phenomenon we call ``image reflection'' that 
each image on the bisimulation collapse of an LLEE chart is actually an LEE sub-chart enforced by its corresponding well-structured looping-back charts of the LLEE chart. We also call it that the image `reflects' the LLEE structure of the LLEE chart. 
\fi
The significance of our idea is that we directly prove $\mcl{H}$ being an LEE chart according to $f$, without performing any graph transformations which rely on careful selections of bisimilar-node pairs. 
\ifx
Our approach advances in that we give a direct proof according to $f$, without 
performing any graph transformations which rely on careful selections of bisimilar-node pairs. 
\fi

From Lemma~\ref{lemma:Image Reflection on LLEE Charts --- critical condition} and the proof of Theorem~\ref{theo:restoration for LLEE charts}, our approach actually induces a new way to look at LEE charts by taking `loops' as basic components and our primary study objects. 
An LEE chart can be seen as a space consisting of a (possibly infinite) set of loops that are orchestrated in a hierarchy manner. 
And to prove a chart satisfying LEE property, we only need to show that all loops are orchestrated properly (through the conditions Lemma~\ref{lemma:Image Reflection on LLEE Charts --- critical condition} (\ref{item:lemma:Image Reflection on LLEE Charts --- critical condition - coverage}), (\ref{item:Image Reflection on LLEE Charts --- critical condition - constrained behaviour}) and (\ref{item:lemma:Image Reflection on LLEE Charts --- critical condition - path})). 
This new angle of understanding based on loops is the key to our proof method. 

\ifx
To prove that a chart satisfies LEE property, the significance of our idea is that 

And to prove a chart satisfying LEE property, we only need to show, from a global view, that all loops are orchestrated properly (through the conditions Lemma~\ref{lemma:Image Reflection on LLEE Charts --- critical condition} (\ref{item:lemma:Image Reflection on LLEE Charts --- critical condition - coverage}), (\ref{item:Image Reflection on LLEE Charts --- critical condition - constrained behaviour}) and (\ref{item:lemma:Image Reflection on LLEE Charts --- critical condition - path})). 
This new angle of understanding based on loops is the key to our proof method. 
\fi

\ifx
the images on the bisimulation collapse of an LLEE chart actually reflect (some part of) the LLEE structure of the LLEE chart, in the sense that 
the behaviour of each image is constrained by its corresponding well-structured looping-back charts. 
\fi

\ifx
We propose a different approach by utilizing a phenomenon we call ``image reflection'' that a certain type of images --- well-structured images --- on the bisimulation collapse of an LLEE chart actually reflect (some part of) the LLEE structure of the LLEE chart, in the sense that 
the behaviour of each image is constrained by its corresponding looping-back chart. 
\fi

\section{Future Work}
\label{section:Future Work}
Previous work further extended LEE/LLEE charts to characterize regular expressions modulo bisimulation equivalence by introducing a silent (without doing any action) \emph{1-transition} $X\xrightarrow{1} Y$ (cf.~\cite{Grabmayer20-2,Grabmayer2022}). 
The LEE/LLEE charts with additional 1-transitions, called \emph{LEE/LLEE 1-charts}, 
are defined the same as shown in Sect.~\ref{section:LEE-LLEE Charts}. 
With 1-transitions, a similar correspondence between regular expressions and provable solutions of LLEE 1-charts can be built, analog to 
Prop.~\ref{prop:from 1-free regular expressions to LLEE charts} and ~\ref{prop:from LLEE charts to 1-free regular expressions} for LLEE charts. 
The connection of LLEE 1-charts to the completeness of \Mil\ were fully explored in recent work~\cite{Grabmayer2022}. 


A direct application of our proof to the case of LEE/LLEE 1-charts for proving the completeness of \Mil\ will fail. 
There are mainly two reasons:
(i) We see that the properties Lemma~\ref{lemma:Image Reflection on LLEE Charts --- critical condition} (2) and (3) generally do not hold when 1-transitions are involved;
(ii) There exist bisimulation collapses of LLEE 1-charts that 
cannot be simply altered into an LLEE 1-chart by introducing auxiliary 1-transitions. 
\cite{Grabmayer2022,Grabmayer2023}~gave a counter-example showing (ii). 
Though currently we have found that for some particular bisimulation collapses, 
Lemma~\ref{lemma:Image Reflection on LLEE Charts --- critical condition} (2) and (3) can be restored by introducing auxiliary 1-transitions, 
reason (ii) tells us a restoration of this type for all bisimulation collapses is not possible. 

However, the idea in our proof method can be tailored to simplify the graph transformations proposed in~\cite{Grabmayer2022}, which is quite technically complex and difficult. 
To do this, the primary is to investigate when $\mcl{H}$ is not a bisimulation collapse, how our method can be adapted to work out on a bisimulation relation $\R:V_\mcl{G}\to V_{\mcl{H}}$ from $\mcl{G}$ to $\mcl{H}$. 
Our future work might focus on this aspect.  

\ifx
The extension of our idea to the case of LEE/LLEE 1-charts for proving the completeness of \Mil\ will not be trivial. 
There are mainly two reasons:
(i) We see that the properties Lemma~\ref{lemma:Image Reflection on LLEE Charts --- critical condition} (2) and (3) generally do not hold when 1-transitions are involved;
(ii) There exist bisimulation collapses of LLEE 1-charts that 
cannot be simply altered into an LLEE 1-chart by auxiliary 1-transitions. 
\cite{Grabmayer2022,Grabmayer2023}~gave a counter-example showing this. 
(ii) is more of an essential reason. 
Although we already have in mind a technique to restore properties Lemma~\ref{lemma:Image Reflection on LLEE Charts --- critical condition} (2) and (3) for a certain type of bisimulation collapses. 
Because of the reason (ii), a restoration for all bisimulation collapses will not be possible. 
More complex graph transformations are needed, for example, the transformations that add additional but necessary bisimilar nodes (and transitions) from the bisimulation collapses in order for our restoration technique to work.  
Our future work may focus on these aspects. 
\fi

\section{Acknowledgements}
\ifx
We thank Prof. Zhiming Liu for his financial support to the authors which makes this work possible. 
Also thank Prof. Chao Wang, students Wei Zhang and Quan Sun for their attendances of our discussions. 
\fi

This work is partially supported by the Youth Project of National Science Foundation of China (No. 62102329) and
the Project of National Science Foundation of Chongqing (No. cstc2021jcyj-bshX0120).

%
%
%
%



\bibliographystyle{splncs04}
\bibliography{main}

\end{document}